\newif\iflong
\newif\ifshort

\longtrue

\iflong
\else
\shorttrue
\fi
\documentclass[11pt,a4paper]{article}
\usepackage{amsmath,amssymb,graphicx,amsthm,dsfont}

\usepackage{color,soul}
\usepackage{xspace}
\usepackage{bbm}
\usepackage{mathtools}
\usepackage{paralist}
\usepackage{booktabs}
\usepackage{multirow}
\usepackage{comment}

\usepackage{subcaption}
\usepackage[ruled,vlined,linesnumbered]{algorithm2e}
\usepackage{todonotes}

\usepackage[margin=1in]{geometry}

\SetCommentSty{mycommfont}

\usepackage[numbers]{natbib}

\newcommand{\Agent}{\text{\normalfont Agent}}

\newcommand{\decprob}[3]{
  \smallskip
  
  {
    \centering
    \begin{minipage}{0.9\linewidth}%
      \textsc{#1}\\[0.2ex]
      \textbf{Input:} #2\\[0.2ex]
      \textbf{Question:} #3
    \end{minipage}%
  \par} 
\medskip

}

\usepackage{tikz}
\usetikzlibrary{decorations,arrows,petri,topaths,backgrounds,shapes,positioning,fit,calc,decorations.pathreplacing,patterns,intersections,decorations.pathmorphing,matrix}
\tikzstyle{normalline} = [line width=.8pt]

\makeatletter
\newcommand{\gettikzxy}[3]{%
  \tikz@scan@one@point\pgfutil@firstofone#1\relax
  \edef#2{\the\pgf@x}%
  \edef#3{\the\pgf@y}%
}
\makeatother

\usepackage{cleveref}
\newtheorem{corollary}{Corollary}
\newtheorem{lemma}{Lemma}

\newtheorem{observation}{Observation}
\newtheorem{proposition}{Proposition}
\newtheorem{construction}{Construction}
\newtheorem{theorem}{Theorem}

\theoremstyle{definition}

\newtheorem{example}{Example}
\newtheorem{drule}{Reduction rule}

\crefname{table}{Table}{Tables}
\crefname{figure}{Figure}{Figures}
\crefname{theorem}{Theorem}{Theorems}
\crefname{definition}{Definition}{Definitions}
\crefname{corollary}{Corollary}{Corollaries}
\crefname{observation}{Observation}{Observations}
\crefname{lemma}{Lemma}{Lemmas}
\crefname{example}{Example}{Examples}
\crefname{reduction}{Reduction}{Reductions}
\crefname{construction}{Construction}{Constructions}
\crefname{subsection}{Section}{Sections}
\crefname{section}{Section}{Sections}
\crefname{proposition}{Proposition}{Propositions}
\crefname{algorithm}{Algorithm}{Algorithms}
\crefname{drule}{Reduction rule}{Reduction rules}
\crefname{claim}{Claim}{Claims}

\newcommand{\calF}{\mathcal{F}}
\newcommand{\calC}{\mathcal{C}}
\newcommand{\calB}{\mathcal{B}}

\newcommand{\Pot}{\mathcal{P}}
\newcommand{\SR}{\textsc{Stable Roommates}\xspace}
\newcommand{\SRB}{\textsc{Min-Block-Pair Stable Roommates}\xspace}
\newcommand{\SRA}{\textsc{Min-Block-Agents Stable Roommates}\xspace}

\newcommand{\ESR}{\textsc{Egalitarian Stable Roommates}\xspace}

\newcommand{\ESM}{\textsc{Egalitarian Stable Marriage}\xspace}

\newcommand{\SM}{\textsc{Stable Marriage}\xspace}
\newcommand{\IS}{\textsc{Independent Set}\xspace}
\newcommand{\MIS}{\textsc{Multi-Colored Independent Set}\xspace}

\newcommand{\ETH}{Exponential Time Hypothesis}
\newcommand{\citallmatchingssamesize}{\cite[Chapter~4.5.2]{GusfieldIrving1989}}

\newcommand{\cheapE}{\ensuremath{E^{\textsf{zero}}}}
\newcommand{\expE}{\ensuremath{E^{\textsf{exp}}}}

\newcommand{\expF}{\ensuremath{\calF^{\textsf{exp}}}}
\newcommand{\cheapedge}{zero edge}
\newcommand{\expedge}{costly edge}
\newcommand{\expensive}{costly}

\newcommand{\fineedge}{harmless edge}
\newcommand{\fine}{harmless\xspace}
\newcommand{\critical}{critical\xspace}
\newcommand{\criticaledge}{critical edge}
\newcommand{\culprit}{culprit}

\newcommand{\pref}{\ensuremath{\succ}}

\newcommand{\bp}{\ensuremath{\beta}} % #blocking pairs
\newcommand{\ba}{\ensuremath{\eta}} % #blocking agents
\newcommand{\egalcost}{\ensuremath{\gamma}}

\newcommand{\egalcostn}{egalitarian cost}

\usepackage{xcolor}
\definecolor{dargray}{rgb}{0.18, 0.18, 0.18}
\definecolor{darkgreen}{rgb}{0.01,0.6,0.1}
\definecolor{lightrose}{rgb}{0.996,0.75,0.793}
\definecolor{rose}{cmyk}{0.75, 0.75, 0,0}
\definecolor{winered}{rgb}{0.6,0.1,0.1}
\definecolor{darkyellow}{rgb}{.99, .87, 0.04}
\definecolor{lightyellow}{rgb}{1, 1, 0.6}
\definecolor{transparent}{rgb}{1,1,1}
\definecolor{lightlightgray}{rgb}{0.88, 0.88, 0.88}
\definecolor{lightgray}{rgb}{0.8, 0.8, 0.8}
\definecolor{lightblue}{rgb}{0.527,0.805,0.977}
\definecolor{lightgreen}{rgb}{.74,1,0}
\usepackage{graphicx}

\newcommand{\rank}{\mathsf{rank}}
\newcommand{\marked}{\mathsf{marked}}
\newcommand{\unmarked}{\mathsf{unmarked}}
\newcommand{\first}{\mathsf{first}}
\newcommand{\last}{\mathsf{last}}

\newcommand{\esrtiesfptrunningtime}{\ensuremath{2^{O(\egalcost^3)}\cdot n^3 \cdot (\log n)^3}}

\tikzstyle{blueline} = [thick, blue, dotted]
\tikzstyle{redline} = [thick, red, dashed]
\tikzstyle{blackline} = [thick, black]

\newcommand{\appsymb}{$\star$}

\begin{document}
\sloppy

%%%%%%%%%%%%%%%%%%%%%%%%%%%%%%%%%%%%%%%%%%%%%%%%%%%%%%%%%%%%%%%%%%%%%%%%%%
%%%%%%%%%%%%%%%%%%%%%%%%%%%%%%%%%%%%%%%%%%%%%%%%%%%%%%%%%%%%%%%%%%%%%%%%%%
\title{% \iflong
% \textsc{\Huge Appendix: Full Version}\\
% \fi
How hard is it to satisfy (almost) all roommates?%
%Stable Roommates:\\ Egalitarian is Easy whereas Almost Stable is Hard% 
%Fixed Parameter Tractable
\thanks{This work is supported by the People Programme (Marie Curie Actions) of the European Union's Seventh Framework Programme (FP7/2007-2013) under REA grant agreement number {631163.11} and
Israel Science Foundation (grant no. 551145/14).}}

\author{Jiehua Chen \and Danny Hermelin \and 
  Manuel Sorge \and Harel Yedidsion\\
{\small Ben-Gurion University of the Negev, Beer Sheva, Israel}\\
{\small \texttt{jiehua.chen2@gmail.com, hermelin@bgu.ac.il,}}\\
{\small \texttt{sorge@post.bgu.ac.il, yedidsio@post.bgu.ac.il}}}

\date{}

\maketitle
\thispagestyle{empty}

\begin{abstract}  
\looseness=-1  The classic \SR problem (which is the non-bipartite generalization of the well-known \SM problem) asks whether there is a \emph{stable matching} for a given set of agents, \emph{i.e.} a partitioning of the agents into disjoint pairs such that \emph{no} two agents induce a \emph{blocking pair}.
% meaning that they prefer to be with each other rather than with their assigned partners). 
  Herein, each agent has a \emph{preference list} denoting who it prefers to have as a partner,
  and two agents are blocking if they prefer to be with each other rather than with their assigned partners.
  
Since stable matchings may not be unique, we study an NP-hard optimization variant of \SR, called \ESR, which seeks to find a stable matching with a minimum \egalcostn~$\egalcost$, \emph{i.e.}\ the sum of the dissatisfaction of the agents is minimum. 
The \emph{dissatisfaction} of an agent is the number of agents that this agent prefers over its partner if it is matched; otherwise it is the length of its preference list. We also study almost stable matchings, called \SRB, which seeks to find a matching with a minimum number~$\bp$ of blocking pairs. 
Our main result is that \ESR parameterized by $\egalcost$ is fixed-parameter tractable,
while \SRB parameterized by $\bp$ is W[1]-hard, even if the length of each preference list is at most five.
\end{abstract}

\newpage
\setcounter{page}{1}
\pagestyle{plain}

\section{Introduction}
\label{sec:intro}
%%%%%%%%%%%%%%%%%%%%%%%%%%%%%%%%%%%%%%%%%%%%%%%%%%%%%%%%%%%%%%%%%%%%%%%%
%%%%%%% Section: Introduction
%%%%%%%%%%%%%%%%%%%%%%%%%%%%%%%%%%%%%%%%%%%%%%%%%%%%%%%%%%%%%%%%%%%%%%%%%

This paper presents algorithms and hardness results for two variants of the \SR problem, a well-studied generalization of the classic \SM problem. Before going into describing our results, we give a brief background that will help motivate our work.

% \subsection{\SM and \SR}
\paragraph*{\SM and \SR}

\looseness=-1 An instance of the \SM problem consists of two disjoint sets of $n$ men and $n$ women (collectively called \emph{agents}), who are each equipped with his or her own personal \emph{strict} preference list that ranks \emph{every} member of the opposite sex. The goal is to find a bijection, or \emph{matching}, between the men and the women that does not contain any blocking pairs. A \emph{blocking pair} is a pair of man and woman who are not matched together but both prefer each other over their own matched partner. A matching with no blocking pairs is called a \emph{stable matching}, and \emph{perfect} if it is a bijection between \emph{all} men and women. %The goal in the \SM problem is to compute a perfect stable matching in an input as above.

\SM is a classic and fundamental problem in computer science and applied mathematics, and as such, entire books were devoted to it~\cite{GusfieldIrving1989,Knuth1976,RothSotomayor1992,Manlove2013}. The problem emerged from the economic field of matching theory, and it can be thought of as a generalization of the \textsc{Maximum Matching} problem when restricted to complete bipartite graphs. The most important result in this context is the celebrated Gale-Shapley algorithm~\cite{GaleShapley1962}: This algorithm computes in polynomial time a perfect stable matching in \emph{any} given instance, showing that regardless of their preference lists, there always exists a perfect stable matching between any equal number of men and women.

The \SM problem has several interesting variants. First, the preference lists of the agents may be \emph{incomplete}, meaning that not every agent is an acceptable partner to every agent of the opposite sex. In graph theoretic terms, this corresponds to the bipartite incomplete case. The preference lists could also have \emph{ties}, meaning that two or more agents may be considered equally good as partners. Finally, the agents may not be partitioned into two disjoint sets, but rather each agent may be allowed to be matched to any other agent. This corresponds 
to the non-bipartite case in graph theoretic terms,
%to the case where the acceptability graph is not necessarily bipartite, 
and is referred to in the literature as the \SR problem.

While \SM and \SR seem very similar, there is quite a big difference between them in terms of their structure and complexity. For one, any instance of \SM always contains a stable matching (albeit perhaps not perfect), even if the preference lists are incomplete and with ties.
Moreover, computing some stable matching in any \SM instance with $2n$ agents can be done in $O(n^2)$ time~\cite{GaleShapley1962}. However, an instance of \SR may have no stable matchings at all, even in the case of complete preference lists without ties (see the third example in Figure~\ref{fig:example}). Furthermore, when ties are present, deciding whether an instance of \SR contains a stable matching is NP-complete~\cite{Ronn1990}, even in the case of complete preference lists.

%\ifshort
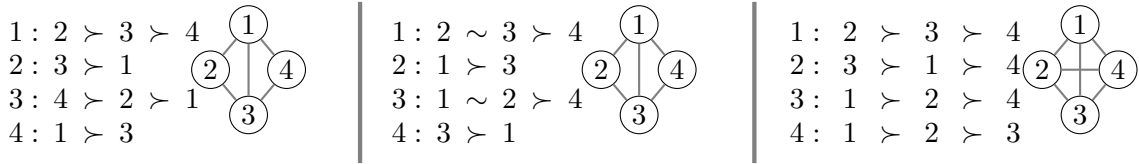
\begin{figure}
  \begin{tikzpicture}
    \def \xw {0.5}
    \def \yw {.6}
    \def \inner {.4pt}
    \def \mins {3ex}
    \foreach \i/\x/\y in {1/0/\yw, 2/-\xw/0, 3/0/-\yw, 4/\xw/0} {
      \node[draw, circle, minimum size=\mins, inner sep=\inner] at (\x,\y)  (\i){$\i$};
    }

     \foreach \i/\j/\w/\v in {1/2/1/2,1/3/2/3,1/4/3/1,2/3/1/2,3/4/1/2}
    {
      \draw[thick] (\i) edge[gray] % node[midway, fill=white,pos=0.2,inner sep=2pt,text=black] {\footnotesize $\w$}
      % node[midway, fill=white, pos=.75,inner sep=2pt,text=black] {\footnotesize $\v$}
      (\j);
    }
     \matrix[below = 0pt of 3] (profile) [matrix of math nodes, row sep=-2pt, column sep=-2pt] at (-1.9,.9) {
    1 : & 2  & \succ & 3 & \succ & 4\\
    2 : & 3 & \succ & 1 \\
    3 : & 4 & \succ & 2 & \succ & 1\\
    4 : & 1 & \succ & 3 \\
    };
  \begin{scope}[xshift=31ex]
    \foreach \i/\x/\y in {1/0/\yw, 2/-\xw/0, 3/0/-\yw, 4/\xw/0} {
      \node[draw, circle, minimum size=\mins, inner sep=\inner] at (\x,\y)  (\i){$\i$};
    }
     \foreach \i/\j/\w/\v in {1/2/1/1,1/3/1/1,1/4/2/2,2/3/2/1,3/4/2/1}
    {
      \draw[thick] (\i) edge[gray] % node[midway, fill=white,pos=0.2,inner sep=2pt,text=black] {\footnotesize $\w$}
      % node[midway, fill=white, pos=.8,inner sep=2pt,text=black] {\footnotesize $\v$}
      (\j);
    }
    \matrix[below = 0pt of 3] (profile) [matrix of math nodes, row sep=-2pt, column sep=-2pt] at (-2,.9){
    1 : &  2  & \sim & 3 & \succ & 4\\
    2 : & 1 & \succ & 3 \\
    3 : & 1 & \sim & 2 & \succ & 4\\
    4 : & 3 & \succ & 1 \\
  };
  \draw[gray, ultra thick] ($(profile.north west)+(-.15,0)$) edge ($(profile.south west)+(-.15,0)$);
  \end{scope}
  \begin{scope}[xshift=66ex]
    \foreach \i/\x/\y in {1/0/\yw, 2/-\xw/0, 3/0/-\yw, 4/\xw/0} {
      \node[draw, circle, minimum size=\mins, inner sep=\inner] at (\x,\y)  (\i){$\i$};
    }
    \foreach \i/\j/\w/\v in {1/2/1/2,1/3/2/1,1/4/3/1,2/3/1/2,2/4/3/2,3/4/3/3}
    {
      \draw[thick] (\i) edge[gray] % node[midway, fill=white,pos=0.2,inner sep=2pt,text=black] {\footnotesize $\w$}
      % node[midway, fill=white, pos=.85,inner sep=2pt,text=black] {\footnotesize $\v$}
      (\j);
    }
    \matrix[below = 0pt of 3] (profile) [matrix of math nodes, row sep=-2pt, column sep=1pt] at (-2.3,.9) {
    1 : &  2  & \succ & 3 & \succ & 4\\
    2 : & 3 & \succ & 1 & \succ & 4 \\
    3 : & 1 & \succ & 2 & \succ & 4\\
    4 : & 1 & \succ & 2 & \succ & 3 \\
    };
  \draw[gray, ultra thick] ($(profile.north west)+(-.2,0)$) edge ($(profile.south west)+(-.2,0)$);
  \end{scope}
\end{tikzpicture}
\caption{An example of three \SR instances, where $x \succ y$ means that $x$ is strictly preferred to $y$,
  and $x\sim y$ means that they are equally good and tied as a partner. The instance on the left is incomplete without ties and has exactly two stable matchings $\{\{1,2\}, \{3,4\}\}$ and $\{\{1,4\},\{2,3\}\}$, both of which are perfect. The instance in the middle is incomplete with ties and has two stable matchings $\{\{1,3\}\}$ and $\{\{1,2\},\{3,4\}\}$, the latter being perfect while the former not. The right instance is complete without ties and has no stable matchings at all.}
\label{fig:example}
\end{figure}

All variants of \SM and \SR mentioned here have several applications in a wide range of application domains. These include partnership issues in the real-world~\cite{GaleShapley1962}, resource allocation~\cite{AbrCheKummir2006,ChenSoenmez2002,HyZe1979},
centralized automated mechanisms that assign children to schools~\cite{AbPaRo2005a,AbPaRo2005b}, assigning school graduates to universities~\cite{BaBa2004,BiKi2015}, assigning medical students to hospitals~\cite{NResidentMatchingP,SResidentMatchingP}, and several others~\cite{AbBiMa2005,GaLeMamoReVi2007,Irving2016a,Irving2016b,KuLiMa1999,LeMaViGaReMo2006,Manlove2008,Manlove2013,MaOM2014,RoSoUn2005,RoSoUn2007}. %For an extensive survey of previous work, the reader is referred to~\cite{}.

% \subsection{Optimization variants}\label{subsec:variants}
\paragraph*{Optimization variants}\label{subsec:variants}

As noted above, some \SR instances do not admit any stable matching at all, and in fact, empirical study suggests that a constant fraction of all sufficiently large instances will have no solution~\cite{PitIrv1994}. Moreover, even if a given \SR instance admits a solution, this solution may not be unique, and there might be other stable matchings with which the agents are more satisfied overall. Given these two facts, it makes sense to consider two types of optimization variants for \SR: In one type, one would want to compute a stable matching that optimizes a certain social criteria in order to maximize the overall satisfaction of the agents. In the other, one would want to compute matchings which are as close as possible to being stable, where closeness can be measured by various metrics. In this paper, we focus on one prominent example of each of these two types---minimizing the egalitarian cost of a stable matching, and minimizing the number of blocking pairs in a matching which is close to being stable.

\subparagraph*{Egalitarian optimal stable matchings.} Over the years, several social optimality criteria have been considered, yet arguably one of the most popular of these is the egalitarian cost metric~\cite{McVWil1971,Knuth1976,IrLeGu1987,MaIrIwMiMo2002,MarxSchlotter2010}. 
The \emph{egalitarian cost} of a given matching is the sum of the \emph{rank}s of the partners of all agents,
%where the \emph{rank} of an ~$x$ in a stable matching depends on %$\rank_x(y)$, 
%the \emph{rank} of its partner~$y$ in this matching, 
where the \emph{rank} of the partner~$y$ of an agent~$x$ is the number of agents that are strictly preferred over $y$ by~$x$. %The \emph{egalitarian cost} of this matching is then the sum of the ranks of the partners of all agents. %in the matching. %, \emph{i.e.} $\sum_{\{x,y\} \in M} \rank_x(y) + \rank_y(x)$. 
The corresponding \ESM and \ESR problems ask whether there is a stable matching with egalitarian cost at most $\egalcost$, for some given bound~$\egalcost \in \mathds{N}$ (\cref{sec:defi} contains the formal~definition).

When the input preferences do not have ties (but could be incomplete), \ESM is solvable in $O(n^4)$~time~\cite{IrLeGu1987}. For preferences with ties, \ESM becomes NP-hard~\cite{MaIrIwMiMo2002}. Thus, already in the bipartite case, it becomes apparent that allowing ties in preference lists makes the task of computing an optimal egalitarian matching much more challenging. \citet{MarxSchlotter2010} showed that \ESM is fixed-parameter tractable when parameterized by the parameter ``sum of the lengths of all ties''.

For \ESR, \citet{Feder1992b} showed that the problem is NP-hard even if the preferences are complete and have no ties, and gave a 2-approximation algorithm for this case. \citet{HaIrIwMaMiMoSc2003} showed inapproximability results for \ESR, and \citet{TeoSet2000} proposed a specific LP formulation for \ESR and other variants.  \citet{CsIrMa2016} studied \ESR for preferences with bounded length~$\ell$ and without ties. They showed that the problem is polynomial-time solvable if $\ell=2$, and is NP-hard for $\ell \geq 3$.
%\todo[inline]{ms: I felt that these paragraphs were a bit aimless and not conducive of our story. We don't excite the emotions of the reader towards wanting to know the parameterized complexity. Can we add something along the lines, that we don't expect the sum  of the length of the ties to be small and that, because of NP-hardness we need other small parameters to get efficient exact algorithms and that---outrageously!---the complexity wrt.\ the standard parameter hasn't been determined yet? (Even short preference lists don't help.)}

\subparagraph*{Matchings with minimum number of blocking pairs.} For the case where no stable matchings exist, the agents may still be satisfied with a matching that is close to being stable. One very natural way to measure how close a matching is to being stable is to count the number of blocking pairs~\cite{NieRot2004,EriHag2008}. Accordingly, the \SRB problem asks to find a matching with a minimum number of blocking pairs.

\citet{AbBiMa2005} showed that \SRB is NP-hard, and cannot be approximated within a factor of $n^{0.5-\varepsilon}$ unless P${}={}$NP, even if the given preferences are complete. They also showed that the problem can be solved in $n^{O(\bp)}$~time, where $n$ and $\bp$ denote the number of agents and the number of blocking pairs, respectively.
This implies that the problem is in the XP class (for parameter~$\bp$) of parameterized complexity. 
\citet{BiMaMcD2012} showed that the problem is NP-hard and APX-hard even if each agent has a preference list of length at most~$3$, and presented a $(2\ell-3)$-approximation algorithm for bounded list length~$\ell$. \citet{BiMaMi2010} and \citet{HaIwMi2009} showed that the related variant of \SM, where the goal is to find a matching with minimum blocking pairs among all maximum-cardinality matchings, cannot be approximated within $n^{1-\varepsilon}$ unless P${}={}$NP.
%\todo[inline]{ms: Ditto. Can we say that, apparently, approximation fails and parameterized algorithms would be nice?} 

% \subsection{Our contributions}\label{subsec:results}
\paragraph*{Our contributions}\label{subsec:results}

We analyze both \ESR and \SRB from the perspective of parameterized complexity, under the natural parameterization of each problem (\emph{i.e.}\ the egalitarian cost and number of blocking pairs, respectively). We show that while the former is 
%(besides a few non-standard variants) 
fixed-parameter tractable, the latter is W[1]-hard even when each preference list has length at most five and has no ties. This shows a sharp contrast between the two problems: Computing an optimal egalitarian stable matching is a much easier task than computing a matching with minimum blocking~pairs.

When no ties are present, an instance of the \ESR problem has a lot of structure, and so we can apply a simple branching strategy for finding a stable matching with egalitarian cost of at most $\egalcost$ 
in $2^{O(\egalcost)}n^2$ time. 
Moreover, we derive a kernelization algorithm, obtaining a polynomial problem kernel (\cref{thm:linear-preference-kernel,thm:esr-no-ties-fpt}).
Note that the original reduction of \citet{Feder1992b} already shows that
\ESR cannot be solved in $2^{o(\egalcost)}n^{O(1)}$ time unless the \ETH~\cite{CyFoKoLoMaPiPiSa2015}~fails.

When ties are present, the problem becomes much more challenging because several
agents may be tied as a first ranked partner and it is not clear how to match them to obtain an optimal egalitarian stable matching. %\todo{ms: We haven't said directly that there are two issues. $\leadsto$ ``A second issue'' or ``Another issue''?}
Moreover, we have to handle unmatched agents. When preferences are complete or without ties, all stable matchings match the same (sub)set of agents and this
subset can be found in polynomial time~\citallmatchingssamesize. Thus, unmatched agents do not cause any real difficulties. However, in the
case of ties and with incomplete preferences, stable matchings may involve different sets of unmatched agents. Aiming at a socially optimal
egalitarian stable matching, we consider the cost of an unmatched agent to be the length of its preference list~\cite{MarxSchlotter2010}. % (also see \cref{sec:defi} for the formal definition). 
(For the sake of completeness, we also consider two other variants where the cost of an unmatched agent is either zero or a constant value, 
and show that both these variants are unlikely to be fixed-parameter tractable.)
Our first main result is given in the following theorem:
\newcommand{\thmESR}{%
\ESR can be solved in \esrtiesfptrunningtime~time,
even for incomplete preferences with ties, 
where $n$ denotes the number of agents and $\egalcost$ denotes the egalitarian~cost.}
\begin{theorem}
\label{thm:ESR}%
\thmESR
\end{theorem}

The general idea behind our algorithm is to apply random separation~\cite{CCC06} to ``separate'' irrelevant pairs from the pairs that belong to the solution matching, and from some other pairs that would not block our solution. This is done in two phases, each involving some technicalities, but in total the whole separation can be computed in $\egalcost^{O(\egalcost)}\cdot n^{O(1)}$ time. After the separation step, the problem reduces to \textsc{Minimum-Weight Perfect Matching}, and we can apply known techniques. 
Recall that for the case where the preferences have no ties, a simple depth-bounded search tree algorithm suffices~(\cref{thm:esr-no-ties-fpt}).

\medskip

In \cref{sec:SRB}, % through a \todo{ms: Show don't tell -> remove ``quite involved''?}quite involved construction,
we show that
\SRB is W[1]-hard with respect to the parameter~$\bp$ (the number of the blocking pairs)  
even if each input preference list has length at most five and does not have ties. 
This implies that assuming bounded length of the preferences does not help in designing an $f(\bp)\cdot n^{O(1)}$-time algorithm for \SRB, unless FPT${}={}$W[1].
Our W[1]-hardness result also implies as a corollary a lower-bound on the running time of any algorithm. %that an algorithm with running time~$n^{o(\bp)}$ is unlikely to exist unless the Exponential Time Hypothesis (ETH) is false.
By adapting our reduction, we also answer in the negative an open question regarding the number of blocking agents proposed by \citet[Chapter~4.6.5]{Manlove2013}~(\cref{cor:srba-w[1]-h}).
\newcommand{\thmSRB}{%
Let $n$ denote the number of agents and $\bp$ denote the number of blocking pairs.
  Even when each input preference list has length at most five and has no ties, \SRB{} is W[1]-hard with respect to $\bp$ %the number~$\bp$ of blocking pairs 
%  even when the preferences are complete and without ties.
  and admits \emph{no} $f(\bp) \cdot n^{o(\bp)}$-time algorithms unless the \ETH\ is false. %, where $n$ denotes the number of agents and $\bp$ denotes the number of blocking pairs.%
}
\begin{theorem}\label{thm:SRB}
\thmSRB
\end{theorem}
\iflong
Our results are summarized in \cref{tab:results}.
\newcommand{\mynew}[1]{{\color{winered}\textbf{#1}\citechsy}}

\newcommand{\citeronn}{$^\triangle$\xspace}
\newcommand{\citeirv}{$^\diamondsuit$\xspace}
\newcommand{\citegi}{$^\spadesuit$\xspace}
\newcommand{\citefeder}{$^\heartsuit$\xspace}
\newcommand{\citechsy}{$^\star$\xspace}
\newcommand{\other}[1]{{\color{gray}#1}}
\newcommand{\citeabm}{$^{\bullet}$\xspace}
\newcommand{\citeirvtwo}{$^{\clubsuit}$\xspace}
\newcommand{\citemiimm}{$^{\square}$\xspace}
\newcommand{\citecim}{$^{+}$\xspace}
\newcommand{\citetan}{$^{\Join}$\xspace}

\begin{table}
\centering
\resizebox{\textwidth}{!}{
  \begin{tabular}{l@{\;}ccc}
    \toprule
    Objective & Without ties &~~~& With ties\\
    \; Parameter \\
    \midrule 

    Any stable matching & $O(n^2)$~\cite{Irving1985}  && NP-c~\cite{Ronn1990} \\[1ex]
    \ESR  &  \multicolumn{1}{c}{NP-c~\cite{Feder1992b}} &&  \multicolumn{1}{c}{NP-c~\cite{Ronn1990}}\\
     \; Egalitarian cost~$\egalcost$ &\\             
    \quad Unmatched agents' costs $=$ pref.\ list length&  \multicolumn{1}{c}{\mynew{$O(2^{\egalcost}\cdot n^2)$, size-$O(\egalcost^2)$ kern.\ [T.~\ref{thm:linear-preference-kernel}+\ref{thm:esr-no-ties-fpt}]}}& &\multicolumn{1}{c}{\mynew{$\egalcost^{O(\egalcost)}\cdot n^{O(1)}$ [T. \ref{thm:ESR}]}}\\
     \quad Unmatched agents' costs $= 0$ &   \multicolumn{1}{c}{\mynew{$O(2^{\egalcost}\cdot n^2)$, size-$O(\egalcost^2)$ kern.\ [T.~\ref{thm:linear-preference-kernel}+\ref{thm:esr-no-ties-fpt}]}}& & \multicolumn{1}{c}{\mynew{NP-h~($\egalcost=0$) [T.~\ref{thm:esr-ties-incomplete-np-hard-const-maxdeg-zero-cost}]}}\\    
     \quad Unmatched agents' costs $=$ a constant &   \multicolumn{1}{c}{\mynew{$O(2^{\egalcost}\cdot n^2)$, size-$O(\egalcost^2)$ kern.\ [T.~\ref{thm:linear-preference-kernel}+\ref{thm:esr-no-ties-fpt}]}}&& \multicolumn{1}{c}{\mynew{W[1]-h [T.~\ref{thm:egal-cost-fixed-w[1]-h}], XP~[P.~\ref{prop:ESR-unmatched-constant-XP}]}}\\[1ex]    
    \SRB
              &  \multicolumn{1}{c}{NP-c~\cite{AbBiMa2005}, NP-c ($\ell=3$)~\cite{BiMaMcD2012}}&&  \multicolumn{1}{c}{NP-c ($\bp=0$)~\cite{Ronn1990}}\\
    \; \#Blocking pairs $\bp$&  \multicolumn{1}{c}{$n^{2\cdot \bp+2}$~\cite{AbBiMa2005}, \mynew{W[1]-h ($\ell=5$) [T.~\ref{thm:SRB}]}} && \multicolumn{1}{c}{NP-c ($\bp=0$)~\cite{Ronn1990}}\\[1ex]
    \textsc{Min-Block-Agents Stable Roommates}
    & \multicolumn{1}{c}{\mynew{NP-h [C.~\ref{cor:srba-w[1]-h}]}} && \multicolumn{1}{c}{NP-c ($\ba=0$)~\cite{Ronn1990}}\\
    \;  \#Blocking agents $\ba$ & \mynew{$O(2^{\ba^2}\cdot n^{\ba+2})$, W[1]-h ($\ell=5$) [C.~\ref{cor:srba-w[1]-h}]} && \multicolumn{1}{c}{NP-c ($\ba=0$)~\cite{Ronn1990}}\\
    \bottomrule
\end{tabular}
}
\caption{Classical and parameterized complexity results of \ESR, \SRB, and \textsc{Min-Block-Agents Stable Roommates}. Herein, $n$ denotes the number of agents. Results marked in {\color{winered}\textbf{bold}} and with~{\color{winered}\citechsy} are obtained in this paper. ``W[1] $(\ell=5)$'' means that the W[1]-hardness (for the respective parameter) holds even if each input preference list has length at most five.}
\label{tab:results}
\end{table}
\fi

% \subparagraph*{Related work.}
\noindent Besides the relevant work mentioned above % in \cref{subsec:variants}
there is a growing body of research regarding the parameterized complexity of preference-based stable matching problems~\cite{MarxSchlotter2010,MarxSchlotter2011,MnichSchlotter2017,MeeksRastegari2017,GuRoSaZe2017,CheNieSko2018}.

\ifshort{}\noindent Due to space constraints we deferred the proofs for results marked by \appsymb\ to an appendix.\fi

%%% Local Variables:
%%% mode: latex
%%% TeX-master: "par_stable"
%%% End:

%%%%%%%%%%%%%%%%%%%%%%%%%%%%%%%%%%%%%%%%%%%%%%%%%%%%%%%%%%%%%%%%%%%%%%%%%
%%%%%%%%%%%%%%%%%%%%%%%%%%%%%%%%%%%%%%%%%%%%%%%%%%%%%%%%%%%%%%%%%%%%%%%%%
\section{Definitions and notations}
\label{sec:defi}
%%%%%%%%%%%%%%%%%%%%%%%%%%%%%%%%%%%%%%%%%%%%%%%%%%%%%%%%%%%%%%%%%%%%%%%%%

%%%%%%%%%%%%%%%%%%%%%%%%%%%%%%%%%%%%%%%%%%%%%%%%%%%%%%%%%%%%%%%%%%%%%%%%%
%%%%%%%%%%%%   section: Definitions and notations
%%%%%%%%%%%%%%%%%%%%%%%%%%%%%%%%%%%%%%%%%%%%%%%%%%%%%%%%%%%%%%%%%%%%%%%%%

\iflong
We introduce necessary concepts and notation for the paper.
\fi
Let $V=\{1,2,\ldots, n\}$ be a set of even number~$n$~agents. Each agent~$i\in V$ has a subset of agents~$V_i\subseteq V$ which it finds \emph{acceptable} as a partner and has a \emph{preference list~$\succeq_i$} on~$V_i$ (\emph{i.e.}\ a transitive and complete binary relation on $V_i$). Here, $x \succeq_i y$ means that $i$ weakly prefers $x$ over $y$ (\emph{i.e.}\ $x$ is better or as good as $y$). We use $\succ_i$ to denote the asymmetric part (\emph{i.e.}\ $x\succeq_i y$ and $\neg (y\succeq_i x)$)
and $\sim_i$ to denote the symmetric part of $\succeq_i$ (\emph{i.e.}\ $x\succeq_i y$ and $y \succeq_i x$). %\todo{ms: Add ``If $x \sim_i y$ then we also say that $x$ and $y$ are \emph{tied} in $i$'s preference list.''?}
For two agents~$x$ and $y$,
we call $x$ \emph{most acceptable} to~$y$ if $x$ is a maximal element in the preference list of $y$. Note that an agent can have more than one most acceptable agent. \ifshort We extend $\succeq$ to $X \succeq Y$ for pairs of disjoint subsets~$X, Y \subseteq V$ in the natural way.\else For two disjoint subsets of agents $X \subseteq V$ and $Y \subseteq V$, $X \cap Y \neq \emptyset$, we write $X\succeq Y$ if for each pair of agents~$x \in X$ and $y \in Y$ we have $x\succeq y$.\fi
%We use $\overrightarrow{X}$ to denote an arbitrary but complete and asymmetric preference list on~$X$.

%\paragraph{Preference profiles and ranks.}
\looseness=-1 A preference profile~$\Pot$ for $V$ is a collection~$(\succeq_i)_{i\in V}$ of preference lists for each agent~$i\in V$. A profile~$\Pot$ may have the following properties: It is \emph{complete} if for each agent~$i\in V$ it holds that $V_i \cup \{i\}= V$; otherwise it is \emph{incomplete}. \iflong The profile~$\Pot$ has \emph{ties} if there is an agent~$i \in V$ for which there are two agents~$x,y \in V_i$ such that $x \sim_i y$ and we say that $x$ and $y$ are \emph{tied} by~$i$.\else If there are three agents~$i \in V$, $x,y \in V_i$ such that $x \sim_i y$, then we say that $x$ and $y$ are \emph{tied} by~$i$ and that the profile~$\Pot$ has \emph{ties}. \fi To an instance $(V,\Pot)$ we assign an \emph{acceptability graph}, which has $V$ as its vertex set and two agents are connected by an edge if each finds the other acceptable. Without loss of generality, $G$ does not contain isolated vertices\iflong, meaning that each agent has at least one agent which it finds acceptable\fi. The \emph{rank} of an agent~$i$ in the preference list of some agent~$j$ is the number of agents~$x$ that $j$ strictly prefers over~$i$: $\rank_j(i)\coloneqq |\{x \mid x \succ_j i\}|\text{.}$

%For an illustration, consider three profiles and their acceptability graphs depicted in \cref{fig:example}.
%The circles labeled with numbers represent the agents and the weak orders next to them represent their preference lists.
%For instance, in the profile of the middle figure, agent~$1$ finds $2$ and $3$ equally good and prefer them both over $4$.
%The ranks of agents~$2,3,4$ are $1,1,3$, respectively.

%\paragraph{Blocking pairs and stable matchings.}
\iflong Given a preference profile~$\Pot$ for a set~$V$ of agents, a \emph{matching}~$M\subseteq E(G)$ is a subset of disjoint pairs~$\{x,y\}$ of agents with $x\neq y$ (or edges in $E(G)$), where $E(G)$ is the set of edges in the corresponding acceptability graph~$G$. % We use the notion~$V(M)$ to denote the set of agents that are assigned a partner by $M$, \emph{i.e.}\ $V(M)\coloneqq \{x,y \mid \{x,y\}\in M\}$, and
For a pair~$\{x,y\}$ of agents, if $\{x,y\}\in M$, then we denote the corresponding partner~$y$ by~$M(x)$; otherwise we call this pair \emph{unmatched}. We write $M(x)=\bot$ if agent~$x$ has \emph{no partner}; \emph{i.e.}\ if agent~$x$ is not involved in any pair in~$M$.
% We say that $M$ is \emph{perfect} if \emph{no} agent~$x$ has $M(x)=\bot$. %it assigns a partner to each agent.
If \emph{no} agent~$x$ has $M(x)=\bot$ then $M$ is \emph{perfect}. %it assigns a partner to each agent. 
\else
For a preference profile with acceptability graph~$G$ and edge set~$E(G)$, a \emph{matching}~$M\subseteq E(G)$ is a subset of disjoint pairs~$\{x,y\}$ of agents with $x\neq y$. % We use the notion~$V(M)$ to denote the set of agents that are assigned a partner by $M$, \emph{i.e.}\ $V(M)\coloneqq \{x,y \mid \{x,y\}\in M\}$, and
If $\{x,y\}\in M$, then we denote the partner~$y$ of $x$ by $M(x)$; otherwise we call the pair~$\{x, y\}$ \emph{unmatched}. We write $M(x)=\bot$ if agent~$x$ has \emph{no partner}; \emph{i.e.}\ if agent~$x$ is not involved in any pair in~$M$.
% We say that $M$ is \emph{perfect} if \emph{no} agent~$x$ has $M(x)=\bot$. %it assigns a partner to each agent.
If \emph{no} agent~$x$ has $M(x)=\bot$ then $M$ is \emph{perfect}. %it assigns a partner to each agent. 
\fi

Given a matching $M$ of $\Pot$, an unmatched pair~$\{x,y\}\in E(G)\setminus M$ \emph{is blocking} $M$ if both~$x$ and $y$ prefer each other to being unmatched or to their assigned partners, \emph{i.e.}\ it holds that $\big(M(x)=\bot \vee y\succ_x M(x) \big) \wedge \big( M(y)=\bot \vee x \succ_y M(y) \big)$.
We call a matching~$M$ \emph{stable} if no unmatched pair is blocking $M$. 
% \iflong Note that this stability concept is called \emph{weak stability}~\citet{Irving1994} when we allow ties in the preferences.
% There are two more stability concepts studied in the literature which we do not consider in this paper; for more information we refer to the textbooks by \citet{GusfieldIrving1989} and \citet{Manlove2013}.
% \fi
\iflong
The \SR problem is defined as follows:
\decprob{\SR}
{A preference profile~$\Pot$ for a set~$V=\{1,2,\ldots, n\}$ of $n$ agents with $n$ being even.}
{Does $\Pot$ admit a stable matching?}
\noindent By the definition of stability, a stable matching for an instance with complete preferences must assign a partner to each agent. 
This leads to the following observation.
\else

The \SR problem has as input a preference profile~$\Pot$ for a set~$V$ of (even number)~$n$ agents and asks whether $\Pot$ admits a stable matching.
\fi
\iflong By our definition of stability, a stable matching for complete preferences must assign a partner to each agent. 
\begin{observation}\label[observation]{obs:SRT->perfect}
Each stable matching in a \SR instance with complete preferences is perfect.
\end{observation}
\else
\begin{observation}\label[observation]{obs:SRT->perfect}
When preferences are complete, each stable matching is perfect.
\end{observation}
\fi

The two problems we consider in the paper are \ESR and \SRB. The latter asks to determine whether a given preference profile $\Pot$ for a set of agents $V$ has a stable matching with at most $\bp$ blocking pairs. The former problem asks to find a stable matching with minimum \egalcostn{}; the \egalcostn{} of a given matching~$M$ is as follows:
{$%\displaystyle
  \egalcost(M) \coloneqq \sum_{i\in V}\rank_i(M(i)), $ %+\rank_j(i)) + \sum_{i \in V\setminus V(M)}|V_i|.$
}%
\noindent where we augment the definition~$\rank$ with $\rank_i(\bot)\coloneqq |V_i|$. For example, the second profile in \cref{fig:example} has two stable matchings~$M_1=\{\{1,3\}\}$ and $M_2= \{\{1,2\}, \{3,4\}\}$ with $\egalcost(M_1)=4$ and $\egalcost(M_2)=2$.

\looseness=-1 The \egalcostn{}, as originally introduced for the \SM problem, does not include the cost of an unmatched agent because the preference lists are complete. For complete preferences, a stable matching must assign a partner to each agent (\cref{obs:SRT->perfect}),
meaning that our notion of \egalcostn{} equals the one used in the literature. For preferences without ties, all stable matchings match the same subset of agents~\citallmatchingssamesize. Thus, the two concepts differ only by a fixed value which can be pre-determined in polynomial time~\citallmatchingssamesize. For incomplete preferences with ties, there seems to be no consensus on whether to ``penalize'' stable matchings by the cost of unmatched agents~\cite{CsIrMa2016}. Our concept of \egalcostn{} complies with \citeauthor{MarxSchlotter2010}~\cite{MarxSchlotter2010}, but we tackle other concepts as well~(\cref{subsec:egal-variants}).

\section{Minimizing the egalitarian cost}\label[section]{sec:egal}
\looseness=-1 In this section we give our algorithmic and hardness results for
\ESR. \cref{sec:egal-noties} treats the case when no ties
are present, where we can use a straightforward branching strategy. In
\cref{sec:egal-ties} we solve the case where ties are present. Herein,
we need a more sophisticated approach based on random
separation. Finally, in \cref{subsec:egal-variants}, we study variants
of the \egalcostn, differing in the cost assigned to unmatched agents.

% \todo[inline]{HY: For the opening of the chapter I would start with some motivation or outline which is directly related to our problem and not with the stable marriage reference.}
% For profiles with complete preferences, and without ties, \citet{IrLeGu1987} provides a polynomial-time algorithm to find a stable matching with minimum egalitarian cost for \SM\todo{Why don't we solve this?}.
% When ties and incomplete preferences are allowed, however,
% the problem becomes NP-hard.
% \citet{Feder1992b} showed that \ESR is NP-complete even for linear orders. In this section we show that \ESR\ is fixed-parameter tractable with respect to the egalitarian cost. With ties, we need a somewhat sophisticated approach, whereas if 
% \todo[inline]{HY: Maybe also address the differences between complete and incomplete preferences}
\newcommand{\goodpair}{good pair}
\newcommand{\exppair}{costly pair}

\subsection{Warm-up: Preferences without ties}\label[section]{sec:egal-noties}
By the stability concept, if the preferences have no ties and two agents~$x$ and $y$ that are each other's most acceptable 
\ifshort agents, 
\else agents (\emph{i.e.}\
$\rank_{x}(y)=\rank_{y}(x)=0$), \fi %and consider no other agent as most
%acceptable, 
then any stable matching must contain 
\iflong the
pair~$\{x,y\}$,
\else
$\{x,y\}$,
\fi
which has cost zero. 
Hence, we can safely add such pairs to a solution matching. After we have matched all pairs of agents with zero cost, 
%prefer each other, 
all remaining, unmatched agents induce cost at least one when they are matched. This leads to a simple depth-bounded branching algorithm.
In terms of kernelization, 
we can delete any two agents 
\iflong that are most acceptable to each other (\emph{i.e.}\ they induce zero cost) and delete agents from some preference list that are ranked higher than $\egalcost$.
\else %
that induce zero cost and delete agents from some preference list that are ranked higher than $\egalcost$.
\fi
This gives us a polynomial kernel.

\looseness=-1
First, we recall a part of the polynomial-time algorithm by
\citet{Irving1985} which finds an \emph{arbitrary} stable matching 
\iflong in a \SR instance
without ties if one exists. 
\else
for preferences without ties.
\fi 
The whole algorithm works in two phases. 
We present here a modified version of the first phase 
to determine ``relevant'' agents by sorting out \emph{fixed}
pairs---pairs of agents that occur in every stable
matching~\cite[Chapter 4.4.2]{GusfieldIrving1989}---and \emph{marked
  pairs}---pairs of agents that cannot occur in any stable
matching. 
\iflong
We note that the first phase of \citet{Irving1985}'s algorithm deletes marked pairs, but we will keep them, because they are important for maintaining the egalitarian cost.
\fi
The modified phase-1 algorithm is given in \cref{alg:phase-1}. Herein, by \emph{marking a pair~$\{u,w\}$} we mean marking the agents $u$ and $w$ in the preference lists of $w$ and~$u$, respectively.
%\SetKwProg{Fn}{Function}{}{}
\SetAlFnt{\sffamily}
\renewcommand\ArgSty{\normalfont\sffamily}
\SetAlCapFnt{\normalfont\sffamily\large}

\begin{algorithm}
  \DontPrintSemicolon
 \ifshort \caption{A modified version of the phase-1 algorithm of \citet{Irving1985}.}
 \else \caption{A modified version of the phase-1 algorithm of \citet{Irving1985} (the complete algorithm finds a stable matching in a \SR instance without ties, if one exists).}
 \fi
  \label[algorithm]{alg:phase-1}
  
  \footnotesize
  \Repeat{no new pair was marked in the last iteration}
  {
    \ForEach{agent~$u \in U$ whose preference list contains at least one unmarked agent\label{line:outer_foreach_loop}}
    {
      $w\leftarrow$ the first agent in the preference list of $u$ such that $\{u, w\}$ is not yet marked\label{line:first_assignment_to_w}
      
      \ifshort
      \lForEach{$u'$ with $u\pref_{w} u'$}{
        \textbf{mark} $\{u',w\}$\label{alg:marking}
      }
      \else
       \lForEach{$u'$ with $u\pref_{w} u'$}{
        \textbf{mark} $\{u',w\}$\label{alg:marking}
      }
      \fi
    }  
  }
\end{algorithm}

\iflong
For an illustration, consider the following profile with ten agents.

\begin{example}\label[example]{ex:egal-w/o-ties}
  The following profile, ignoring the underlines of the agents in the preference lists, has ten agents with preference lists that do not have ties but may be incomplete.
  % \begin{alignat*}{2}
  %   \text{agent }1\colon & 5 \pref 2 \pref 6 \pref 4 \pref 7,\quad & \text{agent }5\colon & 2 \pref 3 \pref 1 \pref 8 \pref 7,\\
  %   \text{agent }2\colon & 6 \pref 7 \pref 5 \pref 1 \pref 4, &  \text{agent }6\colon & 3 \pref 1 \pref 7 \pref 2 \pref 8,\\   
  %   \text{agent }3\colon & 7 \pref 5 \pref 6, &  \text{agent }6\colon & 1 \pref 2 \pref 5 \pref 3,\\   
  %   \text{agent }4\colon & 1 \pref 2, &  \text{agent }8\colon & 6 \pref 5.\\   
  % \end{alignat*} 
  \begin{alignat*}{4}
    &\text{agent }1\colon && 6 \pref \underline{2} \pref 7 \pref \underline{4} \pref \underline{10} \pref \underline{3} \pref \underline{5} \pref \underline{8} \pref \underline{9},\quad && \text{agent }6\colon && \underline{2} \pref 3 \pref 1 \pref \underline{8},\\
    &\text{agent }2\colon && 7 \pref 8 \pref \underline{6} \pref \underline{1}, &&  \text{agent }7\colon && \underline{3} \pref 1 \pref 8 \pref 2,\\   
    &\text{agent }3\colon && 8 \pref 6 \pref \underline{1} \pref \underline{7}, &&  \text{agent }8\colon && 2 \pref \underline{6} \pref 7 \pref 3 \pref \underline{1} ,\\  
    &\text{agent }4\colon && \underline{1}, &&  \text{agent }9\colon && \underline{1},\\
    &\text{agent }5 \colon&& 10 \pref \underline{1}, && \text{agent }10\colon&& 5 \pref \underline{1}.
  \end{alignat*}
  Our modified version of the phase-1 algorithm marks a subset of the agents, indicated by the underlines.
  These markings are used to keep track of pairs that do not belong to any stable matching.
  Now, observe that both agents $4$ and $9$ have preference lists that consist only of marked agents.
  By the results of~\citet{GusfieldIrving1989}, we can conclude that no stable matching will assign any partner to agent~$4$ or~$9$.

  Let $p=\{5,10\}$.
  One can verify that the above profile has exactly two stable matchings~$M_1=\{\{1,6\}, \{2,7\}, \{3,8\}, p\}$ and $M_2=\{\{1,7\}, \{2,8\}, \{3,6\}, p\}$.
  Observe that the pair~$p$ exists in every stable matching as they are each other's most acceptable agents that are available for them.
  Hence, by definition $p$ is a fixed pair.
  The egalitarian costs of $M_1$ and $M_2$ are $\egalcost(M_1)=10$ and $\egalcost(M_2)=8$, where the unmatched agents~$4$ and $9$ have each contribute a cost of one which is the length of their preference lists.
\end{example}
\fi
\looseness=-1
Let $\Pot_0$ be the preference profile produced by \iflong the phase-1 algorithm shown in \fi \cref{alg:phase-1}.
\ifshort
We introduce some more notions.
For each agent~$x$, let $\first(\Pot_0, x)$ and $\last(\Pot_0, x)$ denote the first and the last agent in the preference list of $x$ that are not marked, respectively.
We call a pair~$\{x,y\}$ a \emph{fixed pair} if $\first(\Pot_0, x) = y$ and $\first(\Pot_0, y) = x$. 
Let $\marked(\Pot_0)$ denote the set of all agents whose preference lists consist of only marked agents,
and let $\unmarked(\Pot_0)$ denote the set of all agents whose preference lists have at least one unmarked agent.
By \cite[Chapters 4.4.2 and 4.5.2]{GusfieldIrving1989}, we can neglect all agents that are in the fixed pairs and ignore all ``irrelevant'' agents from $\marked(\Pot_0)$. 
\else
For each agent~$x$, let $\first(\Pot_0, x)$ and $\last(\Pot_0, x)$ denote the first and the last agent in the preference list of $x$ that are not marked, respectively.
Then, the following shows that we can ignore irrelevant agents whose preference lists consists of only marked agents.

\begin{proposition}{\cite[Chapters 4.4.2 and 4.5.2]{GusfieldIrving1989}}\label[proposition]{prop:phase-1}
  For each two agents~$x$ and $y$ in the phase-1 profile~$\Pot_0$, the following holds.
  \ifshort
  \begin{inparaenum}[(1)]
  \else\begin{compactenum}[(1)]
  \fi
    \item\label{prop:first-last} If $\first(\Pot_0, x) = y$, then $\last(\Pot_0, y)=x$.
    \item\label{prop:fixed-pairs} If $\first(\Pot_0, x) = y$ and $\first(\Pot_0, y)=x$, 
    then $\{x,y\}$ exists in every stable matching. 
    \item\label{prop:marked} If the preference list of $x$ consists of only marked agents, then no stable matching assigns a partner to $x$.
    \item\label{prop:unmarked} If the preference list of $x$ has some unmarked agent, then every stable matching must assign some partner to agent~$x$.
    \item\label{prop:no-marked-pair} No stable matching contains a pair that is marked.
    \item\label{prop:marked-pair-non-blocking} No marked pair is blocking any matching that consist of only unmarked pairs.
    \ifshort
\end{inparaenum}
\else
\end{compactenum}
\fi
\end{proposition}

By \cref{prop:phase-1}(\ref{prop:first-last}), the \emph{fixed pairs} are those pairs $\{x, y\}$ such that $\first(\Pot_0, x) = y$ and $\first(\Pot_0, y) = x$. 
Observe that the preference lists of $x$ and $y$ cannot obtain any unmarked agent other than each other.
As already discussed in the beginning of the section, any stable matching for preferences without ties must contain all fixed pairs.
We introduce two more notions to partition the agent set~$V$.
Let $\marked(\Pot_0)$ denote the set of all agents whose preference lists consist of only marked agents,
and let $\unmarked(\Pot_0)$ denote the set of all agents whose preference lists have at least one unmarked agent.
Obviously, 
\iflong $\marked(\Pot_0)$ and $\unmarked(\Pot_0)$ partition the set~$V$, \emph{i.e.}\ $V=\marked(\Pot_0)\uplus \unmarked(\Pot_0)$.
\else
$\marked(\Pot_0)$ and $\unmarked(\Pot_0)$ partition the set~$V$.
\fi
%By the above known results, we can partition the agent set into two disjoint subsets~$A\uplus B$ such that $A$ consists of those agents whose preference lists have some unmarked agents, 
%while $B$ consists of those agents whose preference lists consists of only marked agents.
We can now predetermine some special cases.
\fi
\iflong
\ifshort\begin{lemma}[\appsymb]\else
\begin{lemma}\fi
\label[lemma]{lem:ESR-no-instance}
  If $|\marked(\Pot_0)|>\egalcost|$ or $\unmarked(\Pot_0)$ has an agent~$x$ with $\rank_x(\first(\Pot_0,x))>\egalcost$, % in the preference list of $x$ is at least $\egalcost+1$,
   then the original profile for $\Pot_0$ admits no stable matchings with egalitarian cost at most~$\egalcost$.
\end{lemma}

\iflong
\begin{proof}
  Assume that $|\marked(\Pot_0)|>\egalcost$.
  By \cref{prop:phase-1}(\ref{prop:marked}), we know that no stable matching will assign a partner to any agent in $\marked(\Pot_0)$.
  Since each of unmatched agent will contribute a cost of at least one, we deduce that each stable matching has egalitarian cost at least $|\marked(\Pot_0)|$.
  Thus, by our assumption, each stable matching will have egalitarian cost more than $\egalcost$. 

  Assume that $\unmarked(\Pot_0)$ has some agent~$x$ with $\rank_x(\first(\Pot_0, x))> \egalcost$.
  By \cref{prop:phase-1}(\ref{prop:unmarked}) , each agent that has some unmarked agent in its preference list must obtain some partner in every stable matching. 
  Thus, by our assumption that the first unmarked agent in the preference list of $x$ has a rank larger than $\egalcost$, 
  the partner that $x$ obtains from each stable matching will contribute to an egalitarian cost of more than $\egalcost$.
  Consequently, each stable matching must also have an egalitarian cost more than $\egalcost$.
\end{proof}
\fi

\noindent Using Lemma \ref{lem:ESR-no-instance}, we can shrink our instance to obtain a polynomial size problem kernel. %Now, we are ready to show our kernelization algorithm.
\else
We can now shrink our instance to obtain a polynomial size problem kernel.
\fi

\begin{theorem}\label{thm:linear-preference-kernel}
  \ESR\ without ties admits a size-$O(\egalcost^2)$ problem kernel with at most $2\egalcost + 1$ agents and at most $\egalcost+1$ agents in each of the preference lists.
\end{theorem}

\iflong\begin{proof}
\else\begin{proof}[Proof sketch.]
\fi
  Let $I=(\Pot, V, \egalcost)$ be an instance of \ESR{} and let $\Pot_0$ be the profile that \cref{alg:phase-1} produces for~$\Pot$.
  We use $F$ to denote the set of agents of all fixed pairs\iflong (\emph{i.e.}\ agents in pairs $\{x, y\}$ such that $\first(\Pot_0, x) = y$ and $\first(\Pot_0, y) = x$)\fi, and we use $O$ to denote the set of ordered pairs~$(x,y)$ of agents 
such that $x$ ranks $y$ higher than $\egalcost$.
Briefly put, our kernelization algorithm will delete all agents in $F\cup \marked(\Pot_0)$,
  and introduce $O(\egalcost)$ dummy agents to replace the deleted agents and some more that are identified by $O$.  
  Initially, $F$ and $O$ are set to empty sets.

  \begin{compactenum}[1.]
    \item\label{kern:preprocess-no} If
    $|\marked(\Pot_0)| > \egalcost$ or if there is an agent~$x$ in $\unmarked(\Pot_0)$ with $\rank_x(\first(\Pot_0,x)) > \egalcost$, 
    \ifshort
    then return a trivial no-instance.
    \else
    then replace the input instance with a trivial no-instance; otherwise proceed with the remaining steps.
    \fi
    \iflong
    \item\label{kern:favorite} For each two agents~$x,y\in \unmarked(\Pot_0)$ with $\first(\Pot_0,x)=y$ and $\first(\Pot_0,y)=x$, add to $F$ the agents~$x$ and $y$.

    \item \label{kern:egal-cost-update} Update the egalitarian cost bound for the agents that are either in a fixed pair or unmatched by any stable matching.
    Let $\hat{\egalcost} = \egalcost-\sum_{x\in F}\rank_x(\first(\Pot_0,x))-\sum_{x\in \marked(\Pot_0)}|V_x|$.
    \else
    \item\label{kern:favorite,egal-cost-update} For each two agents~$x,y\in \unmarked(\Pot_0)$ with $\first(\Pot_0,x)=y$ and $\first(\Pot_0,y)=x$, add to $F$ the agents~$x$ and $y$.
    % Update the egalitarian cost bound for the agents that are either in a fixed pair or unmatched by any stable matching.
    Let $\hat{\egalcost} = \egalcost-\sum_{x\in F}\rank_x(\first(\Pot_0,x))-\sum_{x\in \marked(\Pot_0)}|V_x|$.
    \fi

    \item\label{kern:preprocess-no-2}
     If $\hat{\egalcost} < |\unmarked(\Pot_0)\setminus F|$, then 
     \ifshort
     return a trivial no-instance.
     \else
     replace the input instance with a trivial no-instance; otherwise do the remaining steps.
     \fi
    \item\label{kern:dummy} Add to the original agent set a set~$D$ of $2k$ \emph{dummy agents}~$d_1,d_2,\dots, d_{2k}$, where $k=2\lceil \hat{\egalcost}/2\rceil$,
  such that for each $i\in \{1,2,\dots, k\}$,
  the preference list of $d_i$ consists of only~$d_{k+i}$, and the preference list of $d_{k+i}$ consists of only~$d_i$.
  \iflong  In this way, each stable matching must contain all pairs~$\{d_i, d_{k+i}\}$, which have zero egalitarian costs.
  \fi
  
    \item\label{kern:ordered-pairs} \iflong For each two agents~$x,y\in \unmarked(\Pot_0)$ with $\rank_x(y) > \hat{\egalcost}$, add to $O$ the ordered pair~$(x,y)$.\else
    For each two~$x,y\in \unmarked(\Pot_0)$ with $\rank_x(y) > \hat{\egalcost}$, add to $O$ the ordered pair~$(x,y)$.
    \fi

    \item\label{kern:update-preference-lists} For each agent~$a\in \unmarked(\Pot_0)\setminus F$ do the following.

    \begin{compactenum}[(1)]
      \item 
      \iflong For each value~$i\in \{0,1,2,\dots, \hat{\egalcost}\}$, let $x$ be the agent with $\rank_a(x)=i$ and do the following.
      If $x\in F\cup \marked(\Pot_0)$ or if $(x,a)\in O$, then replace in $a$'s preference list agent~$x$ with a dummy agent~$d$, using a different dummy agent for each~$i$, and append $a$ to the preference list of~$d$.
      \else For each~$i\in \{0,1,2,\dots, \hat{\egalcost}\}$, let $x$ be the agent with $\rank_a(x)=i$.
      If $x\in F\cup \marked(\Pot_0)$ or if $(x,a)\in O$, then replace in $a$'s preference list agent~$x$ with a dummy agent~$d$, using a different dummy for each~$i$, and append $a$ to the preference list of~$d$.
      \fi
       \item Delete all agents~$y$ in the preference list of $a$ with $\rank_a(y) > \hat{\egalcost}$.
  \end{compactenum}
    \item\label{kern:delete-irrelevant} Delete $F\cup \marked(\Pot_0)$ from $\Pot_0$.
  \end{compactenum}
 \iflong We show that the above algorithm produces a problem kernel with at most $2\hat{\egalcost}+1$ agents and with preference list length at most $\hat{\egalcost}+1$ each.

  First, the correctness of Step~\ref{kern:preprocess-no} is ensured by \cref{lem:ESR-no-instance}.
  Second, by \cref{prop:phase-1}(\ref{prop:fixed-pairs}), we know that each fixed pair must be matched together in each stable matching. That is, if $x$ and $y$ are each other's most preferred unmarked agents (ignoring the marked agents, because they cannot form blocking pairs), then each stable matching must match $x$ and $y$ together.
  Hence, it does not alter the equivalence between the input instance and the kernel if we decrease the egalitarian cost by the amount of ranks of those agents whose partners are fixed.
  Moreover, by \cref{prop:phase-1}(\ref{prop:no-marked-pair}), no agent from $\marked(\Pot_0)$ will have a partner assigned from any stable matching.
  However, we have to take their egalitarian cost into account. 
  This leads to the correctness of Steps~\ref{kern:egal-cost-update} and \ref{kern:preprocess-no-2}.
  From now on, we assume that $\hat{\egalcost}\ge 1$.
  
  The introduction of at most $\hat{\egalcost}+1$ dummy agents in Step~\ref{kern:dummy} does not contribute any egalitarian cost; hence, this step is correct.
 
  In Step~\ref{kern:update-preference-lists}, we update the preference lists of all original agents that will stay in the problem kernel.
  These are those agents that do not belong to $F\cup \marked(\Pot_0)$.
  To see why this step is correct, for each agent~$a\in \unmarked(\Pot_0)\setminus F$,
  we have already reasoned that $a$ will not be assigned a partner from $F\cup \marked(\Pot_0)$.
  Furthermore, to obtain a stable matching with egalitarian cost at most $\hat{\egalcost}$,
  we also cannot assign to an agent~$a\in (\unmarked(\Pot_0)\setminus F)$
  a partner~$x$ such that $(x,a)\in O$ since the rank of $a$ in the preference list of $x$ is higher than~$\hat{\egalcost}$. Note that appending agents to the preference list of dummy agents does not change the fact that each dummy agent can only be matched with another dummy agent in each stable matching. Thus, the first part of Step~\ref{kern:update-preference-lists} is indeed correct.
  Finally, it is obviously correct to remove in the preference list of $a$ all agents~$x$ that have a higher rank: $\rank_a(x)>\hat{\egalcost}$; note that $a\in (\unmarked(\Pot_0)\setminus F)$.
  In this way, the length of the preference list of $a$ is at most $\hat{\egalcost}+1$, and that the preference list of $a$ consists of agents~$x$ with $x\in D\cup (\unmarked(\Pot_0)\setminus F)$.
  
  After all these changes, we delete all agents from $F\cup \marked(\Pot_0)$. Note that their (non-)matches are determined to be the same in each stable matching, and that the corresponding cost has been accounted for by updating the egalitarian cost bound. This shows that the kernelization algorithm is correct. It remains to bound the size of our problem kernel.
  The kernel has exactly $|\unmarked(\Pot_0)\setminus F|$ original agents and at most $\hat{\egalcost}+1$ dummy agents. 
  By Step~\ref{kern:preprocess-no-2}, we know that $|\unmarked(\Pot_0)\setminus F|\le \hat{\egalcost}$.
  Thus, the kernel has at most $2\hat{\egalcost}+1$ agents.
  By Steps~\ref{kern:dummy} and \ref{kern:update-preference-lists}, each of the agents has a preference list of length at most  $\hat{\egalcost}+1$.

  As for the running time, computing $\Pot_0$ takes $O(n^2)$ time and each of the above steps takes $O(n\cdot \egalcost)$ time. 
  Thus, in total, the kernelization algorithm takes $O(n^2)$ time.
\else
The proof that the above algorithm produces a problem kernel with the desired size in the desired running time is deferred to the appendix.
\fi
\end{proof}

\newcommand{\uunderline}[1]{#1}

\iflong
To illustrate our kernelization algorithm, consider \cref{ex:egal-w/o-ties} again and assume that $\egalcost=8$.
By our kernelization algorithm, $F=\{5,10\}$, $\marked(\Pot_0)=\{4,9\}$.
Updating our egalitarian bound, we obtain that $\hat{\egalcost}=6$.
In Step~\ref{kern:dummy}, we need to introduce $2\lceil \hat{\egalcost}/2 \rceil = 6$ dummy agents, $d_1,d_2,\dots, d_6$.
In Step~\ref{kern:ordered-pairs}, we obtain that~$O=\{(1,8), (1,9)\}$.
After Step~\ref{kern:update-preference-lists}, % and Step \cef{kern:delete-irrelevant},
the updated preference lists of all agents in $(\unmarked(\Pot_0)\setminus F)$ could be as follows:
\allowdisplaybreaks
  \begin{alignat*}{4}
    &\text{agent }1\colon && 6 \pref \uunderline{2} \pref 7 \pref d_1 \pref d_2 \pref \uunderline{3} \pref d_3,\quad && \text{agent }6\colon && \uunderline{2} \pref 3 \pref 1 \pref \uunderline{8},\\
    &\text{agent }2\colon && 7 \pref 8 \pref \uunderline{6} \pref \uunderline{1}, &&  \text{agent }7\colon && \uunderline{3} \pref 1 \pref 8 \pref 2,\\   
    &\text{agent }3\colon && 8 \pref 6 \pref \uunderline{1} \pref \uunderline{7}, &&  \text{agent }8\colon && 2 \pref \uunderline{6} \pref 7 \pref 3 \pref d_1\text{,}\\
     & \text{agent }d_1\colon && d_{4} \pref 1 \pref 8, && \text{agent } d_{4}\colon && d_1\text{,}\\
     & \text{agent }d_2\colon && d_{5} \pref 1, && \text{agent } d_{5}\colon && d_2\text{,}\\
     & \text{agent }d_3\colon && d_{6} \pref 1, && \text{agent } d_{6}\colon && d_3\text{.}
  \end{alignat*}
Finally, we delete the agents from $F\cup \marked(\Pot_0)$.
\fi
\looseness=-1
\iflong 

Now, we turn to our simple branching algorithm.
\else
Using a simple branching algorithm, we obtain the following.
\fi
\ifshort\begin{theorem}[\appsymb]\else
\begin{theorem}\fi\label{thm:esr-no-ties-fpt}
  Let $n$ denote the number of agents and $\egalcost$ denote the \egalcostn.
  \ESR\ without ties can be solved in $O(2^{\egalcost} \cdot n^2)$~time. 
\end{theorem}
\iflong
\begin{proof}
  Let $I=(\Pot, V, \egalcost)$ be an instance of \ESR{} and let $\Pot_0$ be the profile that \cref{alg:phase-1} produces for input $\Pot$.
  We aim to construct a stable matching~$M$ of \egalcostn\ at
  most~$\egalcost$ for $V$. 
  
  First, we use $F$ to collect the agents in the fixed pairs, that exist in all stable matchings, \emph{i.e.}\ 
  $F = \{x,y \in \unmarked(\Pot_0)\mid \first(\Pot_0,x)=y \wedge \first(\Pot_0,y)=x\}$.
  We also add to $M$ the corresponding pairs, \emph{i.e.}\
  $M=\{\{x,y\} \subseteq \unmarked(\Pot_0) \mid  \first(\Pot_0,x)=y \wedge \first(\Pot_0,y)=x\}$.
  Second, just as in our kernelization algorithm stated in the proof of \cref{thm:linear-preference-kernel},
  we update our egalitarian bound by setting $\hat{\egalcost} = \egalcost-\sum_{x\in F}\rank_x(\first(\Pot_0,x))-\sum_{x\in \marked(\Pot_0)}|V_x|$.
  
  Our branching algorithm will extend the matching $M$ to find a stable one with \egalcostn\ at most $\hat{\egalcost}$ and works as follows. 
  Pick an arbitrary unmatched agent~$u\in (\unmarked(\Pot_0)\setminus F)$ and let~$V^*_u=\{v\in V_u\setminus F \mid \{u,v\} \text{ is not marked and } \rank_{u}(v) + \rank_{v}(u)\le \hat{\egalcost}\}$ be the set of agents~$v$ which are still acceptable to~$u$ such that $\rank_v(u) + \rank_u(v) \leq \hat{\egalcost}$. Note that
  $|U| \le \hat{\egalcost}+1$ and that, clearly, $u$ cannot be matched to any
  of its acceptable agents outside of~$U$ as otherwise the egalitarian cost will exceed $\hat{\egalcost}$. 
  Branch into all possibilities to add $\{u, v\}$ to~$M$ for~$v \in U$ and decrease
  the remaining budget~$\egalcost$ accordingly; that is, make one
  recursive call for each possibility. If afterwards
  $\hat{\egalcost} > 0$, then recurse with another yet unmatched agent~$u$. If
  $\hat{\egalcost} = 0$ or there is no unmatched agent anymore, then check
  whether the current matching~$M$ is stable in $O(n^2)$ time. Accept
  if~$M$ is stable and otherwise reject.

  Clearly, in $O(n^2)$ time, we can compute $\Pot_0$, match all fixed pairs, and update the egalitarian cost bound.
  The recursive procedure makes at most $\hat{\egalcost}+1$ recursive
  calls,  and in each of them, the budget is reduced by $1,2,\dots, \hat{\egalcost}, \hat{\egalcost}$, respectively.
  To see why the budgets in the first $\hat{\egalcost}$ calls are updated in this way, we observe that for each agent $u\in \unmarked(\Pot_0)$ and each acceptable agent~$v\in V^*_u$ with $\rank_v(u)\le |V^*_u|-1$, 
  it holds that $\rank_{v}(u)\ge 1$ as otherwise the agent~$v'\in V^*_u$ with rank $|V^*_u|$ would be marked, as ensured by \cref{prop:phase-1}(\ref{prop:first-last}).
  Thus, our branching algorithm has a branching vector~$(1,2,\dots, \hat{\egalcost}, \hat{\egalcost})$, which amounts to a running time of $O(\textsf{call-time}\cdot 2^{\hat{\egalcost}})$, where $\textsf{call-time}$ denotes the running time of each call (see for instance \cite[Chapter 8.1]{Nie06} for some discussion on how to obtain the corresponding running time).
  Since each call can be carried out in $O(n^2)$ time, 
  our algorithm runs in $O(n^2\cdot 2^{\hat{\egalcost}})=O(n^2\cdot 2^{\egalcost})$.
\end{proof}
\fi

\subsection{Preferences with ties}\label[section]{sec:egal-ties}
\looseness=-1 When the preferences may contain ties, 
we can no longer assume that 
if two agents are each other's most acceptable agents, denoted as a \emph{\goodpair}, 
then a minimum \egalcostn{} stable matching would match them together; note that \goodpair{s} do not induce any egalitarian cost.
%if they prefer each other the most. %  prefer each other the
% \cref{rule:delete-fixed-pair} may not be applicable anymore because an agent may consider several agents equally most acceptable.
% In this case, even if two agents consider each other most acceptable,
% implying that they are stable if matched together and will not increase the egalitarian cost,
% we still may not be able to match them together to obtain
This is because their match could force other pairs to be matched together that have large cost.
Nevertheless,
a \goodpair\ will never block any other pair\iflong, % will not induce a blocking pair, %block each other,
\emph{i.e.}
no agent in a \goodpair{} will form with an agent in some other pair a blocking pair.
\else. \fi
It is straightforward to see that each stable matching must contain a \emph{maximal} set of disjoint \goodpair{s}.
However, it may also contain some other pairs which have non-zero cost.
We call such pairs \emph{\exppair{s}}.
Aiming to find a stable matching~$M$ with egalitarian cost at most~$\egalcost$,
it turns out that we can also identify in $\egalcost^{O(\egalcost)}\cdot n^{O(1)}$ time a subset~$S$ of pairs of agents, which contains all \exppair{s} of $M$ and contains no two pairs that may \emph{induce a blocking pair}.
It hence suffices to find a minimum-cost maximal matching in the graph induced by~$S$ and the \goodpair{s}. The crucial idea is to use the random separation technique~\cite{CCC06} to highlight the difference between the matched \exppair{s} in $M$ and the unmatched \exppair{s}.
This enables us to ignore the \exppair{s}
which pairwisely block each other or are blocked by some pair in~$M$ 
so as to obtain the desired subset~$S$.  %in such a way that any two pairs in $S$ could be taken into a matching without introducing blocking pairs.
% In this way, we can delete all edges from the corresponding acceptability graph (see \cref{sec:defi} for the definition) that corresponds to neither a \goodpair\ nor a pair from $S$. % with agents as vertices that represents all possible
% %pairs in the solution (any two of which could be taken together),
% By the above reasoning, a minimum cost matching in the graph.
% corresponds to a stable matching with minimum egalitarian cost.

\iflong\subparagraph{Perfectness.} \fi \iflong Before we get to the algorithm, we show that we can focus on the case where our desired stable matching is perfect, \emph{i.e.}, each agent is matched, 
even when the input preferences are incomplete. \else Before describing the algorithm, we show that we can focus on the case of perfect matchings, even for incomplete preferences. \fi
(Note that the case with complete preferences is covered by \cref{obs:SRT->perfect}.)
% If preferences are complete, then each stable matching is perfect, because two unmatched agents would prefer to be with each other. For incomplete preferences, not necessarily each stable matching is perfect. 
We show this by introducing dummy agents to extend each non-perfect stable matching to a perfect one, without altering the \egalcostn.

\ifshort
\begin{lemma}[\appsymb]
  \else
  \begin{lemma}
    \fi
    \label[lemma]{lem:perfect}
  \ESR for $n$ agents and \egalcostn{}~$\egalcost$ is $O(\egalcost \cdot n^2)$-time  reducible to \ESR for at most $n + \egalcost$ agents and \egalcostn~$\egalcost$ with an additional requirement that the stable matching should be perfect.
\end{lemma}
\iflong
\begin{proof}
  Let $(V, \Pot, \egalcost)$ be an instance of \ESR. 
  Construct another instance $(V', \Pot', \egalcost')$ of \ESR\ as
  follows. Define $k = \egalcost$ if $\egalcost$ is even and $k = \egalcost - 1$
  otherwise. Introduce a set $A \coloneqq \{a_1, \ldots, a_{k}\}$ of $k$~agents,
  and let $V' \coloneqq V \cup A$. 
  Let $V^*$ consist of all agents in $V$ that each have at most~$\egalcost$ acceptable agents. To obtain $\Pot'$,
  define the preference list of each agent in $V \setminus V^*$ to be the same as in $\Pot$. 
  All agents in $A$ have the same set of acceptable agents, namely $A\cup V^*$,
  which are tied as most acceptable.
%Each agent~$a_i \in A$, $i = 1, 2, \ldots, \egalcost'$,
 % has as acceptable agents $A \cup V_{\egalcost}$. All of the agents in
 % $A \cup V_{\egalcost}$ are tied in $a_i$'s preference list. 
  Consistently, for each agent~$b \in V^*$, 
  the preference list of $b$ in $\Pot'$ is $L_{b}\succ A$, where $L_b$ is the preference list of $b$ in~$\Pot$.
%of an agent $b \in V_{\egalcost}$, take
 % the preference list of $b$ in $\Pot$ and add in the last position
  %the set $A$ as tied agents. That is, if $b$'s preference list in
  %$\Pot$ is $L$, then it becomes $L \succ A$ in $\Pot'$. 
  This completes the construction of $(V', \Pot', \egalcost)$. It can clearly be carried out in $O(\egalcost \cdot n^2)$~time.

  We claim that $(V, \Pot, \egalcost)$ admits a stable matching~$M$ with
  \egalcostn\ at most $\egalcost$ if and only if $(V', \Pot', \egalcost)$ admits a
  \emph{perfect} stable matching~$M'$ with \egalcostn\ at most $\egalcost$. 

  For the ``only if'' part,  let $V_\bot \subseteq V$ be the agents
  left unmatched by~$M$. Observe that $|V_\bot| \leq \egalcost$ as each unmatched agent has at least one acceptable agent and thus contributes at least one unit to the \egalcostn.
  Moreover, since $|V|$ is even, $|V_\bot|$ is
  even. Construct a matching~$M'$ for $(V', \Pot', \egalcost)$ with
  $M \subseteq M'$ by matching each agent in $V_\bot$ to a unique
  agent in $A$. Match the remaining, so far unmatched, agents in $A$
  among themselves. Note that this is possible because both $V_\bot$
  and $A$ are even.

  Observe that $M'$ is perfect. It is also stable:
  No agent in $V\setminus V_{\bot}$ is involved in a blocking pair according to~$M$
  and each agent in $V_{\bot}$ is matched to some agent in~$A$,
  and each agent in $A$ is matched to one of his most acceptable agents. 
  It remains to determine the \egalcostn\ of~$M'$.

  Note that each agent in $A$ contributes zero units to the \egalcostn{} of~$M'$ because they are matched with their most acceptable agents. 
  Hence, the only difference between $\egalcost(M)$ and
  $\egalcost(M')$ may arise from the cost of the agents in $V_\bot$. 
  Let $b \in V_\bot$ and let $\ell$ be the number of agents
  acceptable to $b$ according to $\Pot$. 
  By our \egalcostn{} definition and by the preference lists of $b$ and $M(b)$ in $\Pot$, 
  the cost of $b$ for $M$ is the same as the cost of $\{b,M'(b)\}$ for $M'$. % by our \egalcostn{} definition.
  % Since $b$ is unmatched by $M$, $b$'s
  % cost contribution to $\egalcost(M)$ is $\ell$. Since $b$ is matched
  % to some $a_i$ in $M'$, and all agents in $A$ are tied in the last position
  % in $b$'s preference list in $\Pot'$, the only agents preferred to
  % $a_i$ by $b$ are those that are acceptable to $b$ according to
  % $\Pot$. Hence, $b$'s cost contributed to $\egalcost(M')$ is $\ell$,
  % the same as to $\egalcost(M)$.
  Hence, indeed $M'$ is a perfect
  stable matching and has \egalcostn{} at most $\egalcost$. % $\egalcost(M') = \egalcost(M) \leq k$.

  For the ``if'' part, let $M'$ be a perfect stable matching of \egalcostn\ at most $\egalcost$
  for~$\Pot'$. 
  Obtain a matching~$M$ for $\Pot$ by taking $M = \{p \in M' \mid p \subseteq V\}$. %, \emph{i.e.}
  %take into $M$ all pairs in $M'$ for which both agents are in~$V$.
  Observe that no two agents~$a, b$ that are both unmatched with respect to $M$ are acceptable to each other as otherwise, they would prefer to be with each other rather than with their respective partners given by $M'$,
  forming a blocking pair for~$M'$ (note that the partners assigned to $a$ and $b$ by $M'$ are in $A$ and hence have a largest rank in the preference lists of $a$ and $b$ according to $\Pot'$).

  We claim that $M$ is stable for $(V, \Pot, \egalcost)$. 
  Suppose, towards a contradiction, that $\{a, b\} \subseteq V$ is blocking $M$. 
  This implies that $a$ and $b$ are acceptable to each other, 
  and by the above reasoning, at least one of the agents~$a$ and $b$ is matched in $M$. 
  Furthermore, either $a$ or $b$ needs to be unmatched by $M$ as,
  otherwise, $\{a, b\}$ is a blocking pair for~$M'$---a contradiction. 
  Say in $M$, agent~$a$ is unmatched but agent~$b$ is matched, 
  implying that $M'(a)\in A$ and $b$ prefers $a$ to its partner~$M(b)=M'(b)$.
  However, %by the definition of $M$, 
  %in the matching~$M'$, the partner of agent~$a$ is from $A$.
  %However, 
  by the definition of the preference list of agent~$a$ in $\Pot'$,
  it prefers agent~$b$ over any agent in~$A$, 
  implying that $\{a, b\}$ is also blocking~$M'$---a contradiction.

  Finally, by a reasoning similar to the one given for the ``only if'' part, we can obtain that the \egalcostn{s} of $M$ and $M'$ remain the same, which is at most $\egalcost$. % Since the agents in
  % $A$ do not contribute any cost to $M'$, the only difference between
  % $\egalcost(M')$ and $\egalcost(M)$ may arise from the cost
  % contributions of the set~$V_\bot$ of agents that are unmatched with
  % respect to~$M$. For each agent $b \in V_\bot$, the contribution of
  % $b$ to $\egalcost(M')$ is its number~$\ell$ of acceptable agents
  % according to~$\Pot$, because it is matched by $M'$ to some agent
  % in~$A$ and the only agents preferred by~$b$ to the agents in $A$
  % according to $\Pot'$ are those that are acceptable to~$b$ according
  % to~$\Pot$. Clearly, the contribution of~$b$ to $\egalcost(M)$ is
  % $\ell$. Thus, indeed, $\egalcost(M) = \egalcost(M') \leq k$.
\end{proof}
\fi

\noindent\looseness=-1\cref{lem:perfect} allows, in a subprocedure of our main algorithm, 
\iflong to compute a minimum-cost \emph{perfect} matching in polynomial time
instead of a minimum-cost \emph{maximal} matching (which is NP-hard).
\else
to compute a min-cost \emph{perfect} matching in polynomial time
instead of a min-cost \emph{maximal} matching (which is NP-hard).
\fi
\subparagraph*{The algorithm.} \looseness=-1 As mentioned, \iflong our algorithm is based on \else we use \fi random separation~\cite{CCC06}.
We apply it already in derandomized form using \citeauthor{Bsh15}'s construction
of cover-free families~\cite{Bsh15}, a notion related to universal sets~\cite{NSS95}. Let $\hat{n}, p, q \in \mathds{N}$ such that $p + q \leq \hat{n}$. A family~$\calF$ of subsets of some $\hat{n}$-element universe~$U$ is called  \emph{$(\hat{n},p,q)$-cover-free family} if for each subset~$S \subseteq U$ of cardinality $p + q$ and each subset $S' \subseteq S$ of cardinality $p$, there is a member~$A \in \calF$ with $S \cap A = S'$.\footnote{%
\iflong 
The standard definition of cover-free families~\cite{Bsh15} is stated differently from but equivalent to ours. Namely, an $(\hat{n}, p, q)$-cover-free family is a tuple $(X,\calB)$, where $\calB$ is a family of $\hat{n}$~subsets of~$X$ 
such that for each list $(B_1, \ldots, B_p) \in {\calB}^{p}$ and each list $(A_1, \ldots, A_q) \in \calB^{q}$ with $B_i \neq A_j$, $i \in \{1, \ldots, p\}, j \in \{1, \ldots, q\}$, we have $\bigcap_{i = 1}^p B_i \not\subseteq \bigcup_{j = 1}^q A_j$. %The size of the family is~$|X|$. 
\citet{BG17} showed that the two definitions are equivalent.
\else
The standard definition of cover-free families~\cite{Bsh15} is stated differently from but equivalent~\cite{BG17} to ours. 
\fi
} The result by \citet[Theorem 4]{Bsh15} implies that if $p \in o(q)$, then 
%there is a small constant~$c \ge 6$ %, %for each constant $c > 1$, 
there is an $(\hat{n},p,q)$-cover-free family of cardinality $q^{O(p)} \cdot \log \hat{n}$
which can be computed in time linear of this cardinality. %, where $e$ denotes Euler's number.
\iflong A similar result with a larger running time is given by \citet[Lemma~1]{chitnis_designing_2016} based on so-called splitters~\cite{NSS95}.\fi

\newcommand{\block}{block}
\newcommand{\phase}{phase}
\newcommand{\Phase}{Phase}
\ifshort
\looseness=-1 In the remainder of this section, we prove \cref{thm:ESR}.
\else
In the remainder of this section, we prove our main result: 
\begingroup
  \def\thetheorem{\ref{thm:ESR}}
  \begin{theorem}\thmESR 
\end{theorem}
\addtocounter{theorem}{-1}
\endgroup
\fi
Let $\Pot$ be a preference profile for a set~$V$ of agents, possibly
incomplete and with ties.
For brevity we denote by a \emph{solution} (of $\Pot$) a stable matching~$M$ with \egalcostn\ at most $\egalcost$. By \cref{lem:perfect}, we assume that each solution is perfect.
 % We say that two
% pairs~$p, p' \in \binom{V}{2}$ \emph{\block} each other if, assuming
% both pairs~$p$ and $p'$ were matched, they would induce a blocking
% pair. Formally, $p$ and $p'$ \emph{\block} each other
% if %supposing $p\coloneqq \{u,u'\}$ and $p'\coloneqq \{v,v'\}$
% %for some labeling of the agents in $p$ as $u, u'$ and the agents in $p'$ as $v, v'$ we have 
% % two agents from $p$ and $p'$ prefer to be with each other than with
% % their partners, i.e.\ if $p = \{u,u'\}$ and $p' = \{v,v'\}$ we have
% $u'\succ_u v \text{ and } u \succ_{u'} v'$, where $p\coloneqq \{u,v\}$ and $p'\coloneqq \{u',v'\}$. 
% In this case, we also say that $p$ and $p'$ block each other due to the blocking pair~$\{u,u'\}$ with $u\in p$ and $u'\in p'$.
%such that $v \succ_u u'$ and $u \succ_v v'$ where $u'$ is the partner
%of $u$ in $p$ and $v'$ the partner of $v$ in~$p'$.
Our goal is to construct a graph with vertex set~$V$ which contains
all matched ``edges'', representing the pairs, of some solution and some other edges for which 
 no two edges in this graph are blocking each other. 
Herein, we say that two edges~$e,e'\in \binom{V}{2}$ are \emph{blocking each other} if, 
assuming
both edges (which are two disjoint pairs of agents) are in the matching, they would \emph{induce a blocking pair},
\emph{i.e.}\ $u'\succ_u v \text{ and } u \succ_{u'} v'$, where $e\coloneqq \{u,v\}$ and $e'\coloneqq \{u',v'\}$. 

Pricing the edges with their corresponding cost, by~\cref{lem:perfect}, it is then enough to find a minimum-cost perfect
matching. The graph is constructed in three \phase{s} (see \cref{alg:egalmatch}).
In the first \phase,
we start with the acceptability graph of our profile~$\Pot$ and remove all edges whose ``costs'' each exceed~$\egalcost$.
In the second and the third \phase{s},
we remove all edges that \block\ each other
while keeping a stable matching with minimum \egalcostn{} intact. 

\begin{algorithm}[t]
  \DontPrintSemicolon
     \footnotesize
  % \SetKwFunction{EgalMatch}{EgalMatch}
  \SetKw{Accept}{accept}
  \SetKw{Reject}{reject}
  
  \KwIn{A set~$V$ of agents, a preference profile~$\Pot$ over~$V$, and a budget~$\egalcost \in \mathds{N}$.}
  
  \KwOut{A stable matching of \egalcostn\ at most~$\egalcost$ if it exists.}
  
  \BlankLine

  \tcc*[l]{\Phase\ 1}
  $(V, E) \leftarrow $ The acceptability graph of $\Pot$\; 
  $\cheapE \leftarrow \{\{x,y\} \in E \mid \rank_x(y)+\rank_y(x) = 0\}$ \tcp*{The set of \cheapedge s in $E$}
  $\expE \leftarrow \{\{x,y\} \in E \mid 1\le \rank_x(y)+\rank_y(x) \le \egalcost\}$ \tcp*{The set of \expedge s in $E$}
  $E_1 \leftarrow \cheapE \cup \expE$\;
  \tcc*[l]{\Phase\ 2}
  $\expF \leftarrow $ $(|\expE|, \egalcost,  \egalcost^3)$-cover-free family over the universe~$\expE$\;\label{ln:ph2-1}
  \ForEach{$E' \in \expF$\label{ln:ph2-2}}{
    % $E_2 \leftarrow E_1 \setminus E'$\;
    Apply \cref{rule:delete-red-edges\iflong,rule:delete-blocking-green-edges\fi} to $E_1$ to obtain $E_2$\;\label{ln:ph2-3}
    \tcc*[l]{\Phase\ 3}
    $\calC \leftarrow $ $(|V|, \egalcost^2 + 2\cdot \egalcost, 2\cdot \egalcost)$-cover-free family over the universe~$V$\;\label{ln:ph3-1}
    \ForEach{$V' \in \calC$\label{ln:ph3-2}}{
      Apply \cref{rule:delete-mismatched-edges,rule:delete-blocking-critical-edges} to~$E_2$ to obtain $E_3$\;\label{ln:ph3-3}
      $M \leftarrow $ Minimum-cost perfect matching in the graph $(V, E_3)$ or $\bot$ if none exists\;\label{ln:ph3-4}
      \lIf{$M \neq \bot$ and $M$ has cost at most~$\egalcost$\label{ln:ph3-5}}{%
        % \Accept
        \Return{$M$}
      }
    }
  }
  \iflong\Reject\fi
  
  % $\unmerge \leftarrow \emptyset$\; 
  % \ForEach{$\mergelt \in \merge$}
  % %
  % {
  %       $\mathcal A \leftarrow \Atomize{\mergelt}$\;
  %       \lIf{$\exists A \in \mathcal A \colon \mu(A) \geq h$}
  %       	{
  %       	$\unmerge \leftarrow \unmerge \cup \mathcal A$
  %       	}
  %       \lElse
  %       	{
  %       	$\unmerge \leftarrow \unmerge \cup \{ \mergelt \}$
  %       	}
  % }
  % \KwRet{$\unmerge$}
  \caption{Constructing a perfect stable matching of \egalcostn\ at most~$\egalcost$.}
  \label{alg:egalmatch}
\end{algorithm}

We introduce some more necessary concepts.
Let $G$ be the acceptability graph corresponding to~$\Pot$ with vertex set~$V$,
which also denotes the agent set,
and with edge set~$E$.
The \emph{cost} of an edge~$\{x,y\}$ is the sum of the ranks of each endpoint in the preference list of the other: $\rank_{x}(y)+\rank_y(x)$. 
% For each edge~$\{x,y\}\in E$, % with $e=\{x,y\}$,
%let $\cost(\{x,y\})$ denote the sum of the ranks of $x$ and $y$ in their partners' preference orders,
%i.e.\
%$\cost(\{x,y\})\coloneqq \rank_{x}(y)+\rank_y(x)$.
We call an edge~$e\coloneqq \{x,y\}$ a \emph{\cheapedge} if it has \emph{cost} zero, \emph{i.e.}\ $\rank_{x}(y)+\rank_y(x)=0$,
otherwise it is a \emph{\expedge} if the cost does not exceed~$\egalcost$.
We ignore all edges with cost exceeding~$\egalcost$.
Note that no such edge belongs to or is blocking any stable matching with egalitarian cost at most~$\egalcost$.
To distinguish between \cheapedge{s} and \expedge{s}, we construct two subsets~$\cheapE$ and $\expE$
such that $\cheapE$ consists of all \cheapedge{s}, % that correspond to \goodpair{s},
\emph{i.e.}\ $\cheapE \coloneqq \{\{x,y\}\in E \mid \rank_{x}(y)+\rank_y(x)=0\}$,
and $\expE$ consists of all \expedge{s}, % that correspond to \exppair{s},
\emph{i.e.}\ $\expE\coloneqq \{\{x,y\} \in E\mid 0<\rank_x(y)+\rank_y(x) \le \egalcost\}$.
%Two edges~$e$ and $e'$ are \emph{blocking each other} if the corresponding pairs block each other.

\ifshort
\smallskip
\noindent \textbf{\Phase\ 1.}
\else
\subparagraph{\Phase\ 1.}
\fi
We construct a graph~$G_1=(V,E_1)$ from $G$ with vertex set~$V$ and with edge set~$E_1\coloneqq \cheapE \cup \expE$.
% we add to $G$ all edges~$\{u, v\}$ such
% that $u$ is the most acceptable agent of~$v$ and at the same time $v$
% is the most acceptable agent of~$u$. We furthermore add all edges
% $\{u, v\}$ such that the dissatisfactions of $u$ with $v$ as a partner
% and vice versa are at most~$k$. Formally, in the first stage we
% construct~$G_1$ with vertex set~$V$ and edge
% set~$E_1 = E_1^\textsf{cheap} \cup E_1^\textsf{exp}$, where
% $E_1^\textsf{cheap} = \{\{u, v\} \mid u, v \in V \wedge \rank_u(v) =
% \rank_v(u) = 1\}$ and
% $E_1^\textsf{exp} = \{\{u, v\} \mid u, v \in V \wedge \rank_u(v) +
% \rank_v(u) - 2 \leq k\}$.
% We assign to each edge $\{u, v\} \in E_1$ the \emph{cost}
% $\rank_u(v) + \rank_v(u)$. Below, we say that the edges in
% $E_1^\textsf{cheap}$ are \emph{cheap} and the edges in
% $E_1^\textsf{exp}$ are \emph{expensive}.
%If $n$ is the number of agents, \emph{i.e.}\ $|V|=n$,
\iflong
We can compute $G_1$ in $O(\egalcost \cdot n^2)$~time with $n$ being the number of agents.
As already noted, no edge with cost exceeding $\egalcost$
belongs to or is blocking any solution.
\fi 
The following is easy to see.

\begin{lemma}\label[lemma]{lem:stage1-preserve}
  If $\Pot$ has a stable matching~$M$ with egalitarian cost at most~$\egalcost$,
  then $M \subseteq E_1$.
\end{lemma}

\noindent Observe also that a \cheapedge\ cannot block any other edge
because the agents in a \cheapedge{} already obtain their most acceptable agents. Thus, we have the following.

\begin{lemma}\label[lemma]{lem:only-expensive}
  If two edges in $E_1$ block each other, then they are both \expedge{s}.
\end{lemma}

% \begin{lemma}
%  If two edges in $E_1$ block each other, then they are both expensive.
% \end{lemma}
% \begin{proof}
%   No cheap edge blocks any other edge, because both agents in a cheap
%   edge are assigned to their most acceptable agent.
% \end{proof}

\newcommand{\worstrank}{\ensuremath{\operatorname{worst\_rank}}}
\ifshort
\noindent \textbf{\Phase\ 2.} 
\else
\subparagraph{\Phase\ 2.}
\fi
In this phase, comprising Lines~\ref{ln:ph2-1}--\ref{ln:ph2-3} in \cref{alg:egalmatch}, we remove from $G_1$ some of the \expedge{s}
that \block\ each other (by \cref{lem:only-expensive},
no \cheapedge{s} are \block{ing} any other edge). %, because both agents in a cheap edge are
%assigned to their most acceptable agent.
%Hence, we can focus on the \expedge{s} in this \phase.
For technical reasons, we distinguish two types of \expedge{s}: %  Use $\worstrank_u(v)$ to denote the sum of the rank of $v$ in $u$'s preference list plus
% the number of agents that are tied with $v$,
% i.e.\ $\worstrank_u(v)\coloneqq \rank_u(v)+|\{x \in V_u\setminus \{v\} \mid  x\sim_u v\}|$. In other words, $\worstrank_u(v)$ is the largest possible rank of $v$ over all linearizations of the preference list of~$u$. 
We say that a \expedge~$e$ with $e\coloneqq \{u,v\}$ is \emph{\critical for its endpoint~$u$} if 
the largest possible rank of $v$ over all linearizations of the preference list of $u$ exceeds $\egalcost$,
\emph{i.e.}
%$\rank_{u}(v)+|\{x\in V_u \setminus \{v\} \mid x\sim_u v\}|> \egalcost$.
$|\{x \in V_u\setminus \{v\} \mid  x\succeq_u v\}| > \egalcost$. %\todo{ms: Am I stupid or should this be $>$ instead of $\ge$? It doesn't matter for the algorithm, but it's inconsistent and unintuitive.}.
%$\worstrank_u(v) \ge k+1$, where $v$ is the other endpoint of~$e$. 
Otherwise,~$e$ is \emph{\fine for~$u$}. If an edge is \critical\ for at least one endpoint, then we call it \emph{\critical} and otherwise \emph{\fine}. Observe that a \critical{} edge could still belong to a solution. If two edges $e$ and $e'$ block each other due to the blocking pair~$\{u, u'\}$ with $u \in e, u' \in e'$ such that $e'$ is \fine\ for $u'$, then we say that $e$ is \emph{\fine{ly} blocking}~$e'$ (at the endpoint~$u'$). Note that \block{ing} is symmetric while \fine{ly} \block{ing} is not.

Intuitively, we want to distinguish the solution edges from all edges blocked by the solution.
There is a ``small'' number of \fine\ edges blocked by the solution, so we can easily distinguish between them. For the \critical\ edges, we do not have such a bound; we deal with the \critical\ edges blocked by the solution in \Phase~3 in some other way.

% ,
% but it may happen that an unbounded number of agents are tied with the other endpoint~$v$,
% i.e.\ the cardinality~$|\{x \in V_u\setminus \{v\} \mid  x\sim_u v\}|$ is not upper-bounded by a function in $k$.

% if $e$ is critical for~$u$,
%implying that the other endpoint is, together with a (possibly)
%unbounded number of other agents, in the worst-possible position of
%$u$'s preference list, so that $e$ could still be contained in a
%stable matching of egalitarian cost at most~$k$.
% Thus we have to tackle \criticaledge{s} separately: Intuitively, we need to distinguish the matched edges from all edges \block ed by $M$.
% Fortunately, the number of \fineedge{s} \block ed by $M$ is upper-bounded by a function in $k$.
% We will show that we can thus distinguish them in FPT time.
% We cannot upper-bound the number \criticaledge{s} in this way, but we can deal with them, in Step 3.%

\ifshort\begin{lemma}[\appsymb]\else
\begin{lemma}\fi\label[lemma]{lem:blocked-bound}
  Let $M$ be a stable matching with \egalcostn\ at most~$\egalcost$.
  In $G_1$, at most $\egalcost^3$ edges are \fine{ly} \block{ed} by some edge in~$M$.
\end{lemma}
\iflong
\begin{proof}
  %Clearly, $M$ has at most~$k$ \expedge{s}.
  By \cref{lem:only-expensive}, if an edge~$e'$ from $G_1$
  is \block{ed} by an edge~$e$ in~$M$, then both edges~$e'$ and $e$ are \expedge{s}.
  Pick a \expedge~$e$ in~$M$, denote its endpoints by $u$ and $v$, and
  let $F_e$ denote the set of all edges that are \fine{ly} \block{ed} by $e$.
  Recall that each edge~$e' \coloneqq \{u',v'\}\in F_e$ that is \block{ed} by~$e$ induces a blocking pair consisting of some endpoint of~$e$, say $u$,
  and some endpoint of~$e'$, say $u'$.
  By the definition of induced blocking pairs, it follows that $\rank_u(u') \le \rank_u(v) - 1$.
  Thus, since $e$ has cost at most $\egalcost$, meaning that $\rank_u(v)+\rank_{v}(u) \le \egalcost$,
  there are at most $\egalcost-2$ different endpoints of the \expedge{s} in $F_e$,
  which each form with an endpoint of $e$ a blocking pair.
  Since each edge~$e'$ in $F_e$ is \fine{ly} blocked by $e$ it follows that
  %if $e$ is \fine{ly} blocking $e'$ due to the pair~$\{u,u'\}$,
  %then 
  $|\{x \in V_{u'}\setminus \{v'\} \mid  x\succeq_{u'}v'\}| \le \egalcost$.
  Thus, in $F_e$, at most $\egalcost$ edges are incident with endpoint~$u'$ and could be \fine{ly} \block{ed} by $e$ due to $\{u',z\}$ with $z\in e$. %and one of the endpoints of~$e$.
  In total, we obtain that $|F_e|\le \egalcost\cdot (\egalcost-2)$.
  Since $M$ has \egalcostn\ at most~$\egalcost$,
  it has at most $\egalcost$ \expedge{s} that could \block\ some other edge (recall that \cheapedge{s} do not block any edges).
  % at most $(\cost(e)-1)$~agents that the endpoints strictly prefer to each other. Hence,
  % there are at most $k - 2$ agents whose incident edges~$e$ could
  % \block. Observe that each agent~$v$ has incident at most~$k$ fine
  % edges, because each agent~$u$ in a fine edge adjacent to~$v$ has
  % rank at most~$k$ in $v$'s preference list, and there is no
  % linearization putting~$u$ in position at least~$k + 1$.
  Hence, there
  are in total $\egalcost \cdot (\egalcost - 2) \cdot \egalcost < \egalcost^3$ edges which are \fine{ly} \block{ed} by some edge in $M$.
\end{proof}
\fi
\noindent
\looseness=-1 
Let $M'\coloneqq M\cap \expE$ be the set of all \expedge{s} in some solution~$M$ and let $B_M$ be the set of all edges \fine{ly} \block{ed} by some edge in~$M$.
By the definition of \expedge{s} and by \cref{lem:blocked-bound},
it follows that $|M'| \leq \egalcost$ and $|B_M| \le \egalcost^3$.
In order to identify and delete all edges in $B_M$ %that are \fine{ly} blocked by some (costly) edge in a solution~$M$
we apply random separation. % Compute a  $(|\expE|, \egalcost + \egalcost^3)$-universal set~$\calF^\textsf{exp}$ over the universe~$\expE$.
Compute a $(|\expE|, \egalcost, \egalcost^3)$-cover-free family~$\calF^\textsf{exp}$ over the universe~$\expE$.
For each member of $\expF$, perform all the computations below (in this \phase\ and in \Phase~3).
%Consider a stable matching~$M$ with egalitarian cost~$k$.
%Let $m \coloneqq |\expE|$ be the number of \expedge{s} in~$G_1$.
By the properties of cover-free families, $\calF^\textsf{exp}$ contains a \emph{good} member~$E'$ that ``separates'' $M'$ from $B_M$,
\emph{i.e.}\ $M' \subseteq E'$ and $B_M \subseteq \expE \setminus E'$.
%Call such a member~$E'$ good for $M$. 
Formally, we call a member~$E' \in \calF^\textsf{exp}$ \emph{good} if there is a solution~$M$  %a stable matching~$M$ with \egalcostn\ at most~$\egalcost$,
%say that an iteration is \emph{good for $M$} if
such that each \expedge{} in $M$ belongs to $E'$,
and each edge that is \fine{ly} \block{ed} by~$M$ belongs to $\expE\setminus E'$.
We also call $E'$ \emph{good for} $M$.
%To illustrate the above idea, for each member~$E ' \in \expF$,
%we will delete some edges from $\expE$ (a similar computation will also be used in \Phase~3)  to obtain a new graph.
%First of all,
% \begin{enumerate}
% \item we call all edges in $E' \cap E_1^\textsf{exp}$ \emph{\color{darkgreen} green} and all remaining edges in $E_1^\textsf{exp}$ \emph{red}, and
% %As we only remove edges below, we cannot create a new solution within the edges of~$G_1$.
% \item we call $E'$ \emph{good}, if there is a stable matching~$M$ with egalitarian cost at most~$k$,
% %say that an iteration is \emph{good for $M$} if
% such that each \expedge{} in $M$ is {\color{darkgreen}green} and each \fineedge{} blocked by~$M$ is {\color{winered}red}.
% \end{enumerate}
By the property of cover-free families, if there is a solution~$M$, then %the \expedge{s} in $M$ and the \fineedge{s} blocked by~$M$ form a subset of~$\expE$ of size at most~$k + k^3$.
%By the property of universal sets,
%there is a member~$E'$ of $\expF$ such that $B_M=E'\cap \expE$ and $B_M = E'$.
$\expF$ contains a member~$E'$ which is good for $M$. %some stable matching $M$ with egalitarian cost at most $k$.
In the following we present two data reduction rules that delete edges and show their correctness. By \emph{correctness} we mean that, if some member~$E' \in \expF$ is good, then the corresponding solution is still present after the edge deletion.

\looseness=-1 Recall that the goal was to compute a graph that contains 
all edges from a solution and some other edges 
such that no two edges in the graph block each other. %  and contain all edges from a solution. %desired stable matching.
Observe that we can ignore the edges in $\expE\setminus E'$, because, if $E'$ is good, then it contains all \expedge{s} in the corresponding solution\iflong; note that $|E'|$ could be unbounded\fi. % are contained in $E'$. we will not ignore the \expensive\ edges of the corresponding solution. 
This implies the correctness of \ifshort{}the first part of \fi{}the following reduction rule. \ifshort{}The correctness for the second part follows from the definition of being good.\fi
\begin{drule}\label[drule]{rule:delete-red-edges}
  Remove all edges in $\expE \setminus E'$ from~$E_1$.
\iflong\end{drule}
%The computation proceeds as follows.
% Thus, we compute graph $G_2$ from $G_1$ by first deleting all edges in $\expE\setminus E'$ from $G_1$.
% that are red by the above definition.
Apply also the following reduction rule.
\begin{drule}\fi\label[drule]{rule:delete-blocking-green-edges}
  % If there are two edges $e, e' \in E'$ that block each other because
  % of a blocking pair~$\{u,u'\}$ with $u\in e$ and $u'\in e'$, if $e$
  % is \fine{} for $u$ and $e'$ is \fine{} for $u'$, then delete both
  % $e$ and $e'$ from $E_1$.
  If there are two edges $e, e' \in E'$ that are \fine{ly} \block{ing} each other, then remove both $e$ and $e'$ from~$E_1$.
\end{drule}
\iflong
\begin{proof}[Proof of the correctness of \cref{rule:delete-blocking-green-edges}]
%  \myqed{correctness of \cref{rule:delete-blocking-green-edges}}
  Let $M$ be a solution for which $E'$ is good.
  Suppose, towards a contradiction, that $e\in M$.
  Since $E'$ is good, no edge is \fine{ly} blocked by $M$ belongs to $E'$.
  Thus, $e'\notin E'$---a contradiction.
  Analogously, we deduce that $e'\notin M$.
 % in the solution~$M$ \fine{ly} block any edge in $E'$. 
 %  Since $B_M \subseteq \expE\setminus E'$, this implies that no two edges in $M$ could \fine{ly} block each other.
  Thus, we can safely delete both $e$~and~$e'$.%contains neither $e$ nor $e'$.
 % If $E'$ is good, then neither $e$ nor $e'$ can be in the
  %corresponding solution~$M$, because, by definition of being good, no
  %edge in~$M$ \fine{ly} \block{s} any edge in~$E'$.
\end{proof}
\fi
% Furthermore, we there are two edges $e, e' \in E'$ that block each other, then 

% Furthermore, for each two edges~$e, e' \in E'$ that block each other

% Observe that, if $E'$ is good, then this does not remove \expensive\ edges of the solution.
\noindent Let $G_2 = (V, E_2)$ be the graph obtained from~$G_1$ by exhaustively applying \cref{rule:delete-red-edges\iflong,rule:delete-blocking-green-edges\fi}. By the goodness of $E'$ and by the correctness of \cref{rule:delete-red-edges\iflong,rule:delete-blocking-green-edges\fi}, we have the following.
% with vertex set~$V$ and edge set~$E_2 = \cheapE
% \cup
% E_2^\textsf{exp}$. Herein, to obtain $E_2^\textsf{exp}$, remove each
% red edge from~$E_1^\textsf{exp}$. Then, for each pair of green
% edges~$e, e' \in E'$, if they \block\ each other with a corresponding
% blocking pair~$u \in e, v \in e'$ and if $e$ is fine for~$u$ and $e'$ is fine
% for~$v$, then remove them both from~$E_1^\textsf{exp}$.
% \begin{lemma}
%   If there is a stable matching~$M$ with egalitarian cost at most~$k$,
%   then, in some iteration, $E_2$ contains all edges of~$M$. Moreover,
%   if there is a stable matching containing only edges of~$E_2$.
% \end{lemma}
% \begin{proof}
%   As $E_2$ is a subset of all possible pairs, clearly, the second part
%   of the lemma holds. For the first part, note that we have removed
%   only expensive edges from~$E_1$ and hence, by
%   \cref{lem:stage1-preserve} only expensive edges could possibly be
%   missing from~$M \cap E_2$. Let $M^\textsf{exp}$ be the set of
%   expensive edges . By the properties of the universal
%   set~$\calF^\textsf{exp}$, there is an
%   element~$E' \in \calF^\textsf{exp}$ such that each expensive edge
%   of~$M$ is green and each edge blocked by~$M$ is red. In the
%   iteration corresponding to~$E'$,
% \end{proof}

% As mentioned before, %$\expF$ contains some member~$E'$ which is good for some solution~$M$. %stable matching~$M$ with egalitarian cost at most $k$.
%Hence,
% if $E'$ is good, then the corresponding solution is preserved in~$G_2$.

\begin{lemma}\label[lemma]{lem:stage2-preserve}
  If there is a stable matching~$M$ with \egalcostn\ at most~$\egalcost$,
  then $\expF$ contains a member~$E'$ such that the edge set~$E_2$ of $G_2$ defined for $E'$ % each expensive edge of~$M$ is green and
  % each fine edge \block ed by~$M$ is red. Consequently,
  contains all edges of~$M$.
\end{lemma}

By \cref{lem:only-expensive} and since all pairs of edges that are \fine{ly} \block{ing} each other are deleted by \cref{rule:delete-blocking-green-edges},
we have the following.
\begin{lemma}\label[lemma]{lem:only-critical}
  If two edges in~$G_2$ \block\ each other due to a blocking pair~$\{u, u'\}$, then one of the edges is \critical{} for~$u$ or~$u'$.
\end{lemma}
% Let $M$ In some iteration, by the property

% The progress we have made in this stage is that, in any minimum-cost perfect
% matching in~$G_2$, if two edges block each other, then one of them
% must be critical. We now show how to deal with such pairs of blocking
% edges.
\ifshort
\noindent \textbf{\Phase\ 3.}
\else
\subparagraph{\Phase\ 3.}
\fi
In \cref{ln:ph3-3} of \cref{alg:egalmatch} we remove from~$G_2$ the remaining (\critical) edges that do not belong to $M$ but are blocked by some other edges.
This includes the edges that are blocked by $M$.
%\criticaledge{s}
%that are blocked by the solution~$M$ or blocking each other. 
% Let $M$ be the solution for which some member $E'\in \expF$ is good and let $G_2$ be the graph with edge set~$E_2$ that is computed in \Phase~2 according to the member~$E' \in \expF$.
% %  In this final \phase, we will delete all \critical{} edges from the graph~$G_2$ that are blocked by some edge from $M$
% % to obtain a new and final graph~$G_3$,
% % where graph~$G_2$ is computed in Step 2 according to $E'$. % good member~$E'\in \expF$.
% %  %This cannot introduce new solutions contained in the edge set
% % %of~$G_3$.
% By \cref{lem:stage2-preserve}, we have that all edges in $M$ are preserved in $G_2$
% and no \fineedge{} that is blocked by some edge in $M$ belongs to $G_2$.
% However, $G_2$ may still contain some \critical{} edges that are blocked by $M$.
% In this final \phase,
% we aim to identify and delete all such \critical{} edges. % $M$ from all \critical edges blocked by~$M$,
%and remove them.
While the number of edges blocked by~$M$ could still be unbounded,
we show that there are only $O(\egalcost^2)$ agents due to which an edge could be blocked by $M$.
The idea here is to identify such agents, helping to find and delete edges blocked by $M$ or blocking some other edges.
%two blocking edges~$e$ and $e'$ with $e\in M$ and $e'$ being \critical. %, where $e$ and $e'$ are blocking each other.
%some blocking pair between~$M$ and \criticaledge{s} is uppder-bounded by $O(\egalcost^2)$.
We introduce one more notion.
Consider an arbitrary matching~$N$ (\emph{i.e.}\ a set of disjoint pairs of agents) of $G_2$.
Let $e \in N$ and $e' \in E_2 \setminus N$ be two edges.
If they induce a blocking pair~$\{u,u'\}$ with $u\in e$ and $u'\in e'$, 
then we say that $u'$ 
%If one endpoint~$v \in e'$ of $e'$ is involved in the corresponding blocking pair,
%then we say that $v$ 
is a \emph{\culprit} of~$N$. %or simply \emph{\culprit} if $N$ is clear from the context.
Similar to the proof of \cref{lem:blocked-bound},
we obtain the following upper bound on the number of \culprit{s} with respect to a solution.
\ifshort\begin{lemma}[\appsymb]\else
  \begin{lemma}\fi\label[lemma]{lem:culprit-bound}
    Let $M$ be a stable matching.
    Then, each culprit of $M$ is incident with some edge in $M$.
    If $M$ has egalitarian cost at most~$\egalcost$, then it admits at most~$\egalcost^2$ \culprit{s}.
\end{lemma}
\iflong
\begin{proof}
  For the first statement, let $u'$ be a \culprit{} of $M$. %let $e\in M$ and $e'\notin E_2\setminus M$ with $e'=\{u',v'\}$
%  such that $e$ and $e'$ are blocking each other due to the blocking pair~$\{u, u'\}$.
  Suppose, towards a contradiction, that $u'$ is not incident with any edge in $M$.
  This means that $u'$ is unmatched under $M$.
  However, by the definition of \culprit{s}, $M$ contains an edge~$e$ such that $u'$ and one of the endpoint of $e$ would form a blocking pair for $M$---a contradiction to $M$ being stable.

  For the second statement, observe that $M$ contains at most $\egalcost$ \expedge{s}.
  Each edge that is \block{ed} by an edge in~$M$ is \expensive\ and is \block{ed} by a \expedge\ in~$M$. 
  Pick a \expedge~$e$ in~$M$ and consider the set~$F_e$ of all edges \block{ed} by~$e$. 
  Recall that each edge~$e'$ \block{ed} by~$e$ would induce a blocking pair~$\{u,u'\}$ with $u\in e$ and $u'\in e'$. 
  By the same reasoning as used for \cref{lem:blocked-bound},
  we obtain that $F_e$ has at most $\egalcost-2$ endpoints.
%  Since $e$ has cost at most~$\egalcost$, 
 % there are in total at most $k - 2$ agents whose incident
  %edges~$e$ could \block. 
%  \todo{ms: Remove verweis to lemma 4, and reinstate same reasoning? It's not longer and spares the reader to look back at lemma 4.}
Hence, $M$ has in total $\egalcost \cdot (\egalcost - 2) < \egalcost^2$
  \culprit{s}.%\todo{Merge this with \cref{lem:blocked-bound}?}.
\end{proof}
\fi

\newcommand{\culset}{\mathsf{CI}}
\newcommand{\zhset}{\mathsf{Z}}
\newcommand{\crset}{\mathsf{R}}
\iflong
\noindent Consider a solution~$M$ and let $\culset(M)=\{v\in V \mid v$ is a culprit of or incident with some  \expedge{} of $M\}$. % be a set consisting of all culprits of $M$ and all vertices that are incident with some \expedge{} of $M$.
\else
\looseness=-1\noindent Consider a solution~$M$ and let $\culset(M)=\{v\in V \mid v$ is a culprit of or incident with some  \expedge{} of $M\}$. % be a set consisting of all culprits of $M$ and all vertices that are incident with some \expedge{} of $M$.
\fi
By \cref{lem:culprit-bound} and since $M$ has at most~$\egalcost$ \expedge{s}, it follows that $|\culset(M)|\le \egalcost^2+2\egalcost$. %the number of \culprit{s} plus the number of agents that are incident with some \expedge\ from the solution is at most $\egalcost^2+2\cdot \egalcost$.
\iflong
Our goal was to delete edges without disturbing the solution~$M$ such that no remaining edges are blocking each other.
To achieve this, we
\else
We aim to
\fi
identify in $\culset(M)$ a subset~$\crset(M)$ of agents incident with a \criticaledge\ in~$M$, \emph{i.e.}\ $\crset(M)=\{v\in \culset(M) \mid \{v,w\} \in M $ with $\{v,w\}$ being \critical{} for $v\}$. 
Since $M$ has at most $2\egalcost$ \expedge{s}, it follows that $|\crset(M)|\le 2\egalcost$.
To ``separate'' $\crset(M)$ from $\culset(M)$,
we compute a $(|V|, 2\egalcost, \egalcost^2 + 2 \egalcost)$-cover-free family~$\calC$ 
\iflong on the set~$V$ of agents.
\else
on the set~$V$.
\fi
\begin{comment}
identify in $\culset(M)$ a subset~$\zhset(M)$ of all agents that are incident with either a \cheapedge{} or a \fineedge{} of $M$, \emph{i.e.}\ $\zhset(M)=\{v\in \culset(M) \mid v$ is incident with a \cheapedge{} or some edge in $M$ that is \fine{} for $v\}$. 
By \cref{lem:culprit-bound}, it holds that $|\zhset(M)|\le \egalcost^2+2\egalcost$ and $|\culset(M)\setminus \zhset(M)|\le 2\egalcost$.
To ``separate'' $\zhset(M)$ from $\culset(M)$,
we compute a $(|V|, \egalcost^2 + 2 \egalcost, 2 \egalcost)$-cover-free family~$\calC$ 
\iflong on the set~$V$ of agents.
By the first statement in \cref{lem:culprit-bound}, if a culprit~$v$ of $M$ does not have an incident edge from $M$ that is \critical{} for $v$, 
then it must be incident with some zero edge or with an edge in $M$ that is \fine{} for $v$.
\else
on the set~$V$.
\fi
\end{comment}
We call a member~$V' \in \calC$ \emph{good} if
there is a solution~$M \subseteq E_2$ such that 
\iflong
for each agent~$v\in \culset(M)$ the following holds.
\begin{compactitem}[-]
\item If $v\in \crset(M)$, then $v \in V'$, \emph{i.e.}\ if $v$ is incident with
an edge in $M$ that is \critical{} for~$v$, then $v \in V'$;
\item otherwise, $v\in V\setminus V'$, meaning that if $v$ is incident with a \cheapedge{} or with an edge in $M$ that is \fine{} for $v$, then $v\in V\setminus V'$. 
\end{compactitem}
\else
$\crset(M)\subseteq V'$ and $(\culset(M)\setminus \crset(M))\subseteq V\setminus V'$.
\fi

\iflong \noindent Since $|\crset(M)|\le 2\egalcost$ and $|\culset(M)\setminus \crset(M)|\le \egalcost^2+2\egalcost$, by a similar reasoning as given for \Phase~2 and by the properties of cover-free families,
\else
By a similar reasoning as given for \Phase~2 and by the properties of cover-free families,
\fi
if there is a solution~$M \subseteq E_2$, then $\calC$ contains a good member~$V'$. %table matching~$M$ with egalitarian cost at most $k$.
\iflong We describe two reduction rules to delete some edges from $G_2$ and show their correctness, \emph{i.e.}\ if some member~$V' \in \calC$ is good, then the rules do not delete any edge of a corresponding solution.%  to obtain our desired $G_3$, depending on the colors of the endpoints of the edges.
\else
For this member, the following two reduction rules will not destroy the solution.
\fi%  to obtain our desired $G_3$, depending on the colors of the endpoints of the edges.
% During the description, we will also reason that, if we are in a partition~$T$ which is good for~$M$,
% then $M$ is preserved in~$G_3$.
% To obtain the edge set~$E_3$ of $G_3$ from the edge set~$E_2$ of $G_2$,
% perform the following edge deletions.
\ifshort\begin{drule}[\appsymb]
\else\begin{drule}\fi\label[drule]{rule:delete-mismatched-edges}
  For each agent~$y \in V\setminus V'$, delete all incident edges
  that are \critical{} for~$y$.
\end{drule}
\iflong
\begin{proof}[Proof of the correctness of \cref{rule:delete-mismatched-edges}]
%  \myqed{correctness of \cref{rule:delete-mismatched-edges}}
  Assume that $V'$ is good and let $M \subseteq E_2$ be a corresponding solution.
  Since $y\notin V\setminus V'$, it follows that $y$ is not incident with an edge in $M$ that is critical for~$y$.
  Thus, we can safely delete all incident edges that are \critical for~$y$.
%   By the first statement of \cref{lem:culprit-bound}   
% If agent~$y$ is incident with an edge in $M$, then this edge is either a \cheapedge\ or an edge which is \fine\ for~$v$.
%   Thus, no deleted edge is contained in $M$.
  % Note that after this operation, $v$ is only incident to \cheapedge{s},
  % which as we already mentioned do not block any other edge.
  % If a blue agent~$b$ is incident with an edge in $M$,
  % then it cannot be \critical{} for~$b$ and thus for each \criticaledge{} incident with $b$, we have $e \notin M$. % \todo{ms: Maybe we can get rid of one of the colors if we say that
%    blue represents ``matched to fine or cheap''?\\ms: I didn't dare touching this atm.}
  % After this operation, no edge that is blocked by some other edge is \critical{} for $u$.  
\end{proof}
\fi
After having exhaustively applied \cref{rule:delete-mismatched-edges}, %has been applied
%exhaustively, 
we use the following reduction rule.
\iflong Recall that by \cref{lem:only-critical}, if two
edges in~$E_2$ block each other, then one of them
is \critical. 
\fi
\ifshort\begin{drule}[\appsymb]
\else\begin{drule}\fi\label[drule]{rule:delete-blocking-critical-edges}
  If $E_2$ contains two edges~$e$ and $e'$ that induce a blocking pair~$\{u,u'\}$ with %due to
  $u \in e$ and $u' \in e'$ such that $e$ is critical for~$u$, then
  remove~$e'$ from~$E_2$.
\end{drule}
\iflong
\begin{proof}[Proof of the correctness of \cref{rule:delete-blocking-critical-edges}]
  Assume that $V'$ is good and let~$M \subseteq E_2$ be a corresponding solution. 
  Towards a contradiction, suppose that $e' \in M$. Hence, $u$ is
  a \culprit\ of~$M$. 
  By the first statement of \cref{lem:culprit-bound}, let $\{u,w\}\in M$ be an incident edge.
  Since \cref{rule:delete-mismatched-edges} has been
  applied exhaustively, all \criticaledge{s} that are incident with
  agents in~$V\setminus V'$ are deleted. 
  Since $V'$ is good, it follows that $u \in V'$ and $\{u,w\}$ is critical for $u$.
  Since $e$ is also critical for~$u$, it follows that $v \sim_u w$, where $v$ is the endpoint of~$e$ different from~$u$.
  However, since $e$ and $e'$ are blocking each other, it follows that $\{u,w\}$ and $e'$ are also blocking each other%
  % Since $u$ is a \culprit, it is incident with an edge~$\{u, w\} \in M$ that is
  % \critical{} for~$u$, \emph{i.e.}\ $|\{x\in V_{u} \setminus \{w\}\mid x \succeq_u w\}| > \egalcost$. Let $v$ denote the
  % other endpoint of~$e$ different from~$u$. It follows that $v \sim_u w$. Since $e$ and $e'$ block
  % each other, $u' \succ_u v$.  It follows that also $u' \succ_u w$.
  % This implies that the matched edge $e'$ blocks another matched edge in $M$
  ---a contradiction to $M$ being stable.
\end{proof}
\fi
Let $G_3 = (V, E_3)$ be the graph obtained after having exhaustively applied \cref{rule:delete-mismatched-edges,rule:delete-blocking-critical-edges} to~$G_2$.
\iflong
As mentioned, by the properties of cover-free families, if there is a solution contained in~$E_2$, then the above constructed cover-free family~$\calC$ contains a good member. Thus, by the correctness of \cref{rule:delete-mismatched-edges,rule:delete-blocking-critical-edges} we have the following.
\else
By the correctness of \cref{rule:delete-mismatched-edges,rule:delete-blocking-critical-edges} we have the following.
\fi
\begin{lemma}\label[lemma]{lem:stage3-preserve}
  If there is a stable matching~$M \subseteq E_2$ with \egalcostn\ at most~$\egalcost$,
  then the constructed cover-free family~$\calC$ contains a good member~$V'\in \calC$ such that the edge set~$E_3$ of $G_3$ resulting from the  application of \cref{rule:delete-mismatched-edges,rule:delete-blocking-critical-edges} contains all edges of~$M$.
\end{lemma}
% \begin{proof}
  % By
% \end{proof}
\noindent Since for each member~$V' \in \calC$, we delete all edges that pairwisely block each other,
each perfect matching in~$G_3$ induces a stable matching. We thus
have the following.
\ifshort\begin{lemma}[\appsymb]
 \else\begin{lemma}\fi\label[lemma]{lem:backwards}
  If $G_3$ admits a perfect matching~$M$ with edge cost at most~$\egalcost$,
  then $M$ corresponds to a stable matching with egalitarian cost at most~$\egalcost$.
\end{lemma}
\iflong
\begin{proof}
  Since the cost of each edge in~$G_3$ is exactly the egalitarian cost induced,
  considering $M$ as a matching, the egalitarian cost of~$M$ is at most~$\egalcost$. 
  Towards a contradiction, suppose that a perfect matching~$M$ in $G_3$ has two blocking
  edges~$e$ and $e'$. % Let $u \in e, v \in e'$ be a corresponding blocking pair. 
  By
  \cref{lem:only-critical} one of $e$ and $e'$ is \critical{} for its
  endpoint in the corresponding blocking pair. 
  Since \cref{rule:delete-blocking-critical-edges} does not apply to $G_3$ anymore,
  the other edge is not present in~$E_3$, a contradiction.
\end{proof}
\fi
\noindent Thus, to complete \cref{alg:egalmatch}, in \cref{ln:ph3-4} we
compute a minimum-cost perfect matching for~$G_3$ and output yes, if it
has \egalcostn\ at most~$\egalcost$.

\looseness=-1 Summarizing, by \cref{lem:perfect} if there is a stable matching of
\egalcostn\ at most~$\egalcost$, then it is perfect and thus, by
\cref{lem:stage1-preserve,lem:stage2-preserve,lem:stage3-preserve},
there is a perfect matching in~$G_3$ of cost at
most~$\egalcost$. Hence, if our input is a yes-instance, then \cref{alg:egalmatch} accepts by returning a desired solution. \iflong Furthermore, if \else If \fi it accepts, then by
\cref{lem:backwards} the input is a yes-instance. 
\ifshort The running time is proved in the appendix.\fi

\iflong
As to the running time, as a slight modification to \cref{alg:egalmatch}, prior to 
any other computation, we first compute an $(n^2, \egalcost, \egalcost^3)$-cover-free family~$\expF$ with cardinality $(\egalcost^3)^{O(\egalcost)}\cdot \log n$ from
\Phase~2 in $O\big((\egalcost^3)^{O(\egalcost)}\cdot \log n\big) = \egalcost^{O(\egalcost)}\cdot \log n$ time~\cite{Bsh15},  % (without loss of generality assume that $\egalcost \geq 3$),
and an $(n, 2\egalcost, \egalcost^2+2\egalcost)$-cover-free family~$\calC$ from \Phase~3 in $O\big((\egalcost^2 + 2\egalcost)^{O(\egalcost)} \cdot \log n \big) = \egalcost^{O(\egalcost)}\cdot \log n$~time~\cite{Bsh15}. % Note that each $(\hat{n}, p, q)$-cover-free family~$\mathcal{F}$ induces an $(\hat{n}, q, p)$-cover-free family by taking the complement of each member in $\mathcal{F}$.
% Thus, $\calC$ has cardinality $(\egalcost^2 + 2\egalcost)^{O(2\egalcost + 1)}\cdot \log n = \egalcost^{O(\egalcost)}\cdot \log n$ and can be computed within time linear in this quantity~\cite{Bsh15}.
Note that we can reuse $\calC$ during the course of the algorithm. The remaining computation time
can be bounded as follows.  First, we compute~$G_1$ in $O(\egalcost\cdot n^2)$
time by checking for each pair of agents, whether they are most
acceptable to each other or whether the sum of their ranks is at most~$\egalcost$.
Then, in \Phase~2, we iterate through all %$O(\egalcost^{4\egalcost}\cdot \log n)$
members of the cover-free family~$\expF$ and for each member we need the
following computation time. We first compute $G_2$ in $O(\egalcost\cdot n^2)$
time (note that $O(\egalcost)$ time is enough to check whether two given edges
block each other, assuming the preference lists are ordered).  Then we
iterate through all members of 
%consider all possible~$\egalcost^{O(\egalcost)}\cdot \log n$ sets in the universal set
the cover-free family~$\calC$ and for each of them
compute~$G_3$. Computing $G_3$ can be done in~$O(\egalcost\cdot n^2)$ time by
similar reasoning as before. Finally, the minimum-cost perfect
matching can be found in $O(n^3 \cdot \log n)$
time~\cite{cook_computing_1999}.  Thus, overall the running time is
$\egalcost^{O(\egalcost)}\cdot n^3 \cdot (\log n)^3$.
Thus we have proved \cref{thm:ESR}.
\fi

\subsection{Variants of the egalitarian cost for unmatched agents}\label[section]{subsec:egal-variants}

As discussed in Sections~\ref{subsec:results} and \ref{sec:defi}, % and in \cref{sec:defi}, 
when the input preferences are incomplete,
a stable matching may leave some agents unmatched.  
%As mentioned before, if the preferences in our input are incomplete,
%some agents may be unmatched in an optimal solution. 
In the absence of ties, all stable matchings leave the same set of agents unmatched~\citallmatchingssamesize. Hence, whether an unmatched agent should infer any cost is not relevant in terms of complexity. 
However, when preferences are incomplete and with ties, 
stable matchings may involve different sets of matched agents.
The cost of unmatched agents changes the parameterized complexity dramatically.
\ifshort
In particular, we find that as soon as the cost of an unmatched agent is bounded by a fixed constant, 
seeking for an optimal egalitarian stable matching is parameterized intractable.
The corresponding results are deferred to the appendix.
\fi

\iflong

In this section, we consider two variants of assigning costs to unmatched agents:
zero cost or a constant fixed cost,
and we show that for both cost variants, seeking for an optimal egalitarian stable matching is parameterized intractable.
(Note that the cost $\infty$, that is, allowing
only perfect matchings, is already covered in \cref{sec:egal-ties}.)

\subparagraph*{Unmatched agents have cost zero.}
By the definition of stability,
if two agents are acceptable to each other, but
their corresponding dissatisfaction is large, then either they are
matched, contributing a large, or possibly too large, portion to the
\egalcostn, or one of the two agents must be matched with someone else.
% elsewhere.
This can be used to model a choice of truth value for a variable in a
reduction from satisfiability problems.
Indeed, we can show that the problem is already hard for \ESM, the bipartite variant of \ESR.

\ifshort\begin{theorem}[\appsymb]\else
  \begin{theorem}\fi\label{thm:esr-ties-incomplete-np-hard-const-maxdeg-zero-cost}
  If the cost of the unmatched agents is zero, then \ESM{} with incomplete preferences and ties is NP-complete even in
  the case where the \egalcostn{} is zero and each agent has at most three acceptable agents.
\end{theorem}
% \begin{figure}
%   \centering
%   % \begin{tikzpicture}
%   %   \tikzstyle{agent} = [draw, circle, minimum size=3.5ex, inner sep=.5pt]
%   %   \foreach \i in {1, 2, 3, 5, 6} {
%   %     \node[agent] at (\i, 0) (a\i) {$a_\i$};
%   %     \node[agent] at (\i, -1) (b\i) {$b_\i$};
%   %     \node[agent] at (\i, -2) (c\i) {$c_\i$};
%   %   }
%   %   \node[agent] at (3.5, 1) (m) {$d_1$};
%   %   \node[agent] at (4.5, 1) (m) {$d_2$};
%   %   \node[agent] at (4, 0) (m) {$d_3$};

%   %   \begin{scope}[xshift=32ex]

%   %   \end{scope}
%   % \end{tikzpicture}
%   \hfill
%   \includegraphics{variable_gadget}
%   \hfill
%   \includegraphics{clause_gadget}
%   \hfill\mbox{}
%   \caption{Variable gadget (left) and clause gadget (right) used in \cref{thm:esr-ties-incomplete-np-hard-const-maxdeg-zero-cost}. Dashed pairs induce cost at least one.}\todo[inline]{Hua:Pls update to use (2,2)-3SAT.}
%   \label{fig:gadgets}
% \end{figure}
\iflong
\newcommand{\true}{\textsf{true}}
\newcommand{\false}{\textsf{false}}
\begin{proof}
  We reduce from the NP-complete \textsc{3SAT} variant in which each literal appears exactly twice and each clause has exactly three literals~\cite{BerKarSco2003}. %with the restriction that each variable
  %occurs in a constant number of clauses.
  Let $\phi$ be a corresponding Boolean
  formula in conjunctive normal form with variable set~$X$ and clause
  set~$\mathcal{C}$.

  The reduction proceeds as follows. For each
  variable~$x_i\in X$, we introduce $3$~agents~$a_i^*$, $a_i^{\true}$, and $a_i^{\false}$,
  and $12$ agents,
  denoted as $b_{i,j}^{s}$, $c_{i,j}^{s}$, and $d_{i,j}^{s}$ for all $j\in \{1,2\}$ and $s\in \{\true, \false\}$.
  \iflong
  The agents for one specific variable~$x_i$ are shown in the variable gadget in \cref{fig:gadgets}.
  \else 
  The agents for one specific variable~$x_i$ are shown in the variable gadget in appendix.
  \fi
  % Herein, $\ell$ is the maximum number
  % of occurrences of a literal in clauses.
  The preference list of an agent~$v$ in the variable gadget is defined as follows: Each agent
  that is adjacent to~$v$ with a solid line is a most acceptable agent of~$v$
  (and vice versa); they have rank~$0$ in each other's preference list. Each agent that is adjacent to~$v$ with a dashed
  line has rank~$1$. 
  For each clause~$C_j$ (note that it contains three literals), we construct a \emph{clause gadget} that consists of
  $2$ agents~$u^*_j$ and $w^*_j$, and of $9$ agents, denoted as $x_{j,i}$, $y_{j,i}$, and $z_{j,i}$ for all $i \in \{1,2,3\}$.
  \iflong
  The agents for clause~$C_j$ are shown in the variable gadget on
  the left in \cref{fig:gadgets}.
  \else 
  The agents for one specific variable~$x_i$ are shown in the variable gadget in appendix.
  \fi
  We use the same conventions (the solid and dashed edges) to define their preference lists as we did for the variable gadgets.

  % the agents in
  % the clause gadget shown on the right in \cref{fig:gadgets} and
  % define their preference lists in the same way as for the variable
  % gadgets.
%
  To combine the clause and variable gadgets, for each
  variable~$x_i \in X$ and for each clause~$C_j\in \mathcal{C}$ we do the following:
  If $x_i$ appears \emph{positively} in a
  clause~$C_j \in \mathcal{C}$, then we pick two agents $c_{i,s}^\true$,
  $d_{i,s}^\true$ in $x_i$'s variable gadget and two agents $y_{j,r}$,
  $z_{j,r}$ in $C_j$'s clause gadget that have not been used for combining
  before, and identify $c_{i,s}^\true = z_{j,r}$ and
  $d_{i,s}^\true = y_{j,r}$.
  Analogously, if $x_i$ appears \emph{negatively} in $C_j$, then we identify $c_{i,s}^{\false} = z_{j,r}$ and $d_{i,s}^{\false} = y_{j,r}$.
  %\todo{Hua: Pls check whether I understand your combination correctly...\\ms: yes}

  \begin{figure}\centering
  \begin{tikzpicture}
    \tikzstyle{nod} = [draw,fill=black,circle, inner sep=1.6pt]
    \node[nod] (a) {};
    \node[above = 0pt of a] (an) {$a_i^*$};

    \node[nod,below = .5 of a, xshift=-1cm] (at) {};
    \node[nod,below = .5 of a, xshift =1cm] (af) {};

    \node[nod,below = .7 of at, xshift=-1.5cm] (bt1) {};
    \node[nod,below = .7 of at, xshift=-.5cm] (bt2) {};

    \node[nod,below = .7 of af, xshift=.5cm] (bf1) {};
    \node[nod,below = .7 of af, xshift=1.5cm] (bf2) {};

    \node[nod,below = .7 of bt1] (ct1) {};
    \node[nod,below = .7 of bt2] (ct2) {};

    \node[nod,below = .7 of bf1] (cf1) {};
    \node[nod,below = .7 of bf2] (cf2) {};

    \node[nod,below = .7 of ct1] (dt1) {};
    \node[nod,below = .7 of ct2] (dt2) {};

    \node[nod,below = .7 of cf1] (df1) {};
    \node[nod,below = .7 of cf2] (df2) {};

    \foreach \o / \n / \i / \j / \m in {left/at//\true/{a_i}, left/bt1/1/\true/{b_{i,1}}, left/ct1/1/\true/{c_{i,1}}, left/dt1/1/\true/{d_{i,1}},
      right/bt2/2/\true/{b_{i,2}}, right/ct2/1/\true/{c_{i,2}}, right/dt2/2/\true/{d_{i,2}},
      right/af//\false/{a_i}, left/bf1/1/\false/{b_{i,1}}, left/cf1/1/\false/{c_{i,1}},
      left/df1/1/\false/{d_{i,1}},
      right/bf2/2/\false/{b_{i,2}}, right/cf2/1/\false/{c_{i,2}}, right/df2/2/\false/{d_{i,2}}%
}
    {
      \node[\o = 0pt of \n] {$\m^{\j}$};
    }

    \foreach \s / \t / \a in {a/at/solid,
      at/bt1/dashed, bt1/ct1/-, ct1/dt1/-,
      at/bt2/dashed, bt2/ct2/-, ct2/dt2/-,
      a/af/solid, af/bf1/dashed, bf1/cf1/-, cf1/df1/-,
      af/bf2/dashed, bf2/cf2/-, cf2/df2/-} {
      \draw[\a, thick] (\s) edge (\t);
    }

    \begin{scope}[xshift=.5\textwidth]

      \foreach \i/\p in {u/.7,w/-.7} {
        \node[] at (\p,-.3) (s\i) {$\i^*_j$};
      }

      \foreach \i  in {u,w} {
        \node[nod, below = 0 of s\i] (\i) {};
    }
    
      \foreach \i / \o / \n in {-1.4/left/1,0/right/2,1.4/right/3} {
        \foreach \x / \j  / \p in {x/1/1.6,y/2/2.4,z/3/3.2} {
          \node[nod] at (\i, -\p) (\x\n) {};
          \node[\o = 0pt of \x\n] {$\x_{j,\n}$};
        }
      }

      \foreach \s in {u,w} {
        \foreach \t in {1,2,3}{
          \draw[thick] (\s) edge (x\t);
        }
      }
      \foreach \s / \t / \a in {x/y/dashed, y/z/solid}
      {
        \foreach \i  in {1,2,3}
        {
          \draw[\a] (\s\i) edge (\t\i);
        }
      }
    \end{scope}
 %    \begin{scope}[xshift=.7\textwidth]

%       \node[nod] (w1) at (0, 0) {};
% %      \node[nod] (w2) at (.7, 0) {};

%         \node[above = 0pt of w1] {$w^*_1$};

%       \foreach \i / \o / \n in {-1/left/1,1/right/2} {
%         \foreach \x / \j  in {x/1,y/2,z/3} {
%           \node[nod] at (\i, -\j) (\x\n) {};
%           \node[\o = 0pt of \x\n] {$\x_\n$};
%         }
%       }

%       \foreach \t in {1,2}{
%         \draw[thick] (w1) edge (x\t);
%       }

%       \foreach \s / \t / \a in {x/y/solid, y/z/dashed}
%       {
%         \foreach \i  in {1,2,3}
%         {
%           \draw[\a] (\s\i) edge (\t\i);
%         }
%       }
%     \end{scope}
  \end{tikzpicture}
  \caption{A variable gadget (left) and a clause gadget (right) for the proof of \cref{thm:esr-ties-incomplete-np-hard-const-maxdeg-zero-cost}.}\label{fig:gadgets}
  \end{figure}
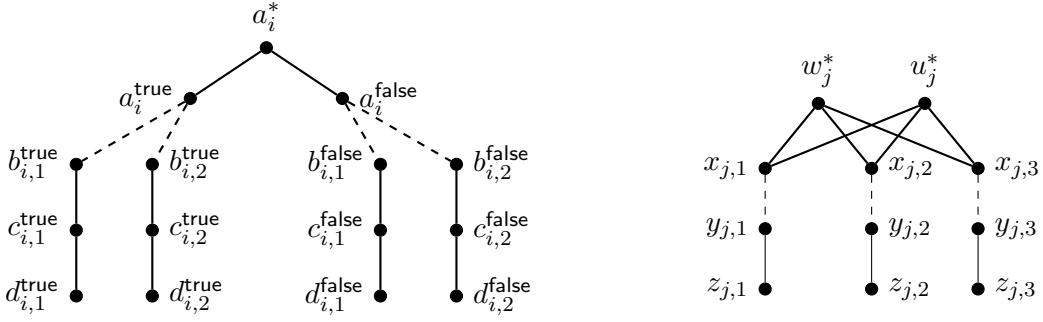

  This completes the construction, which can clearly be done in polynomial time. Observe
  that each agent has \emph{at most three} acceptable agents.
  One can verify that the underlying acceptability graph is bipartite since every cycle has an even length.
  Hence, the constructed preference profile is indeed a valid instance of \SM.

  To see that a satisfying assignment for~$\phi$ induces a stable
  matching with egalitarian cost~$0$, construct a matching~$M$ as follows.
  For each variable~$x_i$ that is assigned to \true,
  let $\{a_i^*, a_i^\true\} \in M$ and, for each $s \in \{1,2\}$, let
  $\{c_{i,s}^\true, d_{i,s}^\true\}, \{b_{i,s}^\false, c_{i,s}^{\false}\}\in M$. Accordingly, for
  each variable~$x_i$ that is assigned to \false, %match
  let $\{a^*, a^\false\}\in M$ and, for each $i \in \{1,2\}$
  let %match
  $\{b_{i,s}^\true, c_{i,s}^\true\}, \{c_{i,s}^\false, d_{i,s}^{\false}\}\in M$.
  % $\{b_i^\textsf{true}, c_i^\textsf{true}\}, \{c_i^\textsf{false}, d_i^{\textsf{false}}\}\in M$.
  Observe that, if an agent from the variable gadget is matched, then it is matched with one of his most acceptable agents. Furthermore,
  if an agent (except $d_{i,s}^{\textsf{false/true}}$) is not matched, then
  all of its acceptable agents are matched. For each clause~$C_j$, pick
  a literal that satisfies it, say it is the literal~$x_1$, and
  let $\{u^*_{j}, x_{j,2}\}, \{w_j^*, x_{j,3}\}\in M$.
  Again, each matched agent is with one of its most acceptable agents, hence the egalitarian cost
  is~$0$. Furthermore, for each agent~$v$ from the clause gadget,
  $v$ is either matched (with \egalcostn~$0$), or each agent acceptable to~$v$ is matched to one of its most acceptable agents. Hence,
  there are no blocking pairs. We conclude that $M$ is a stable matching and has \egalcostn~$0$. %there is a stable matching of
%  egalitarian cost~$0$.

  Now let $M$ be a stable matching of egalitarian cost~$0$. Note that,
  in the variable gadget of each variable~$x_i$, agent~$a_i^*$ is matched to either
  $a_{i}^\true$ or $a_i^\false$.
  This is because $M$ being stable implies that at least one of  $a_i^{\true}$ and $a_i^{\false}$ has to be matched
  and $a^*$ is
  the only agent acceptable to them that has cost~$0$.
% , since
%   $a^\textsf{true}$ and $a^\textsf{false}$ are acceptable to each
%   other, one of the two has to be matched in~$M$ and because $a^*$ is
%   the only agent acceptable to them that induces cost~0.
  For each variable~$x_i$, assign to $x_i$
  \true{} if $\{a_i^*,a_i^\true\}\in M$ and \false{} otherwise. We
  claim that each clause~$C_j \in \mathcal{C}$ is satisfied in this
  way. To see this, consider $C_j$'s clause gadget and observe that one
  of the triple~$x_{j,1}$, $x_{j,2}$, and $x_{j,3}$ is not matched in~$M$ because there are only two agents that are most acceptable to the triple;
  note that the total egalitarian cost should be zero.
  Say $x_{j,1}$ is not matched (the other cases are symmetric). This implies that $p=\{y_{j,1}, z_{j,1}\}\in M$, which corresponds to a literal as by our construction. 
  Assume that this pair~$p$ equals a pair $\{d_{i,1}^{\true}, c_{i,1}^{\true}\}$ for some variable~$x_i$ (\emph{i.e.}\ it occurs positively in~$C_j$).
  This is without loss of generality, due to symmetry. 
  % Observe that, since $x_{j,1}$ is not matched, and
  % $\{x_{j,1}, y_{j,1}\}$ is not blocking $M$,
  % it follows that $y_{j,1}$ is matched to $z_{j,1}$.
  % Assume that $x_{j,1}$ corresponds to the agent~$d_{i,1}^\true$ in $x_i$'s variable gadget.
  % By the connection of the variable and clause gadgets, this means that
  % in $x_i$'s variable gadget $d_{i,1}^\true$ is matched to $c_{i,1}^\true$.
  Since $\{a_i^\true,b_{i,1}^\true\}$ should be neither a member of nor blocking~$M$,
  it follows that $a_i^\true$ is matched to $a_i^*$, meaning that variable~$x_i$ is set to
  \true. Thus, clause~$C_j$ is satisfied, as claimed.
\end{proof}

\subparagraph*{Unmatched agents have some constant positive cost.}
If the unmatched agents have some constant positive cost~$c$, then it
is easy to see that \ESR\ belongs to~XP.

\begin{proposition}\label{prop:ESR-unmatched-constant-XP}
  If the cost of each unmatched agent is some positive constant, then
  \ESR\ with incomplete preferences and with ties can be solved in
  $n^{\egalcost} \cdot \esrtiesfptrunningtime$ time, where $\egalcost$ is the \egalcostn.
\end{proposition}
\begin{proof}[Proof sketch]
  Let $c$ be the cost for an unmatched agent. The algorithm is as
  follows. Guess, by trying all possibilities, a subset~$A$ of at most
  $\egalcost/c$ unmatched agents. Remove $A$ from the set of
  agents~$V$ and all preference lists and modify the preference lists
  of the remaining agents as follows. For each
  agent~$u \in V \setminus A$ who was acceptable to some
  agent~$a \in A$, remove from $u$'s preference list all agents~$b$
  for which $u$ strictly preferred~$a$ over~$b$. In the remaining instance, search for
  a perfect stable matching of egalitarian cost at most~$\egalcost-c\dot |A|$. 
  This can be done in
  \esrtiesfptrunningtime~time using \cref{thm:ESR}. If such a
  matching exists, accept and otherwise reject.
\end{proof}
\fi

\ifshort
\noindent If the cost of each unmatched agent is a positive constant, then \ESR\ admits an $f(\egalcost) \cdot n^{\egalcost + O(1)}$-time algorithm, and we cannot substantially improve on this.
\fi

\iflong
We cannot substantially improve on the above algorithm in general,
however. Indeed, we can use the same idea as in the reduction for
\cref{thm:esr-ties-incomplete-np-hard-const-maxdeg-zero-cost} and
utilize the fact that, when there are ties, an agent can select his partner from an unbounded number of agents with the same cost, in order to obtain a polynomial-time parameterized reduction
from the W[1]-complete \IS{} problem (parameterized by the size of the independent set solution).% to \ESR{}.
\fi

% \todo{adapt reduction below}
\ifshort\begin{theorem}[\appsymb]\else
  \begin{theorem}\fi\label{thm:egal-cost-fixed-w[1]-h}
    Let $n$ denote the number of agents and $\egalcost$ denote the \egalcostn.
  If the cost of each unmatched agent is some positive constant, then
  \ESR\ with incomplete preferences and ties is W[1]-hard with respect to~$\egalcost$.
  It does not admit an $f(\egalcost)\cdot n^{o(\egalcost)}$-time algorithm unless the \ETH\ is false. %, where $n$ denotes the number of agents and $\egalcost$ denotes the \egalcostn.%  NP-complete even in
  % the case where the \egalcostn~$\egalcost$ should be zero.
\end{theorem}
\iflong
\begin{proof}
  We reduce from \IS, which, given an $n$-vertex
  graph~$G = (U, E)$ with vertex set~$U$ and edge set~$E$, and a number~$k\in \mathds{N}$,
  asks whether $G$ admits a $k$-vertex \emph{independent set}, a vertex subset~$U'\subseteq U$ with pairwisely non-adjacent vertices.
  %an integer~$k$ and we want to decide whether
  %there is a size-$k$ vertex subset~$U' \subseteq U$ such that no edge
  %s contained in~$U'$. 
  Let $c$ be the cost for an unmatched agent, which is positive; without loss of generality, we assume that $c$ is also a positive integer. 
  To construct an instance of \ESR, set the \egalcostn{} to
  $c\cdot k$ and construct a preference profile as follows. Set the set~$V$
  of agents to $U \cup A \cup D$, where $U$ is a set of \emph{vertex agents}, $A = \{a_1, a_2, \ldots, a_{n - k}\}$ is a set of $n - k$
  \emph{selector agents}, and
  $D = \{d_{1, i}^u, d_{2, i}^u \mid i \in \{1, 2, \ldots, c\} \wedge u \in U\}$ is a set of $2\cdot c \cdot n$ \emph{dummy agents}.
  Note that we use $U$ to denote both the vertex set and the set of vertex
  agents, and we will make it clear whether we mean the vertices or
  their corresponding vertex agents.

  For each $i \in \{1, 2, \ldots, n - k\}$ the set of acceptable agents to $a_i$ is
  precisely~$U$ and each pair of acceptable agents are tied, that is,
  they have rank~$0$ in $a_i$'s preference list. For each
  vertex~$u \in U$, the set of the acceptable agents of the corresponding agent~$u$ is $A \cup \{d_{1, i}^u \mid i \in \{1, 2, \ldots, c\} \} \cup N(u)$.
  The preference list of $u$ is $A \succ d_{1, 1}^u \succ d_{1, 2}^u \succ \ldots \succ d_{1, c}^u \succ N(u)$,
  \emph{i.e.}\ $A$ are all tied with rank~$0$, followed by $c$ dummy agents, and finally the agents from the neighborhood $N(u)$ are all tied with rank~$c + 1$. Finally, for each $u \in U$ and $i \in \{1, 2, \ldots, c\}$, agent~$d_{1,i}^{u}$ is the only acceptable agent of $d_{2,i}^{u}$, and $d_{1,i}^{u}$'s preference list is $d_{2, i}^{u} \succ u$.
%the preference list of $d_{1, i}^u$ is $d_{2, i}^u \succ u$ and the preference list of $d_{2, i}^u$ consists solely of $d_{1, i}^u$. 
 % has as acceptable
 %  agents all agents in~$A$ as rank~$0$ and all agents in its
 %  neighborhood~$N(v)$ in~$G$ as rank~$2$.
  This completes the construction, which can clearly be done in polynomial time. 
  If the reduction is correct, then it implies the result, because \IS\ is well-known to be W[1]-hard and, moreover, an $f(k) \cdot n^{o(k)}$-time algorithm for \IS would contradict the \ETH~\cite{CyFoKoLoMaPiPiSa2015}.

  To see that a size-$k$ independent set~$U'$ induces a stable
  matching of \egalcostn~$c\cdot k$, match each agent~$a_i$ with a distinct
  vertex in~$U \setminus U'$ and, for each $u \in U$ and
   each $i \in \{1, 2, \ldots, c\}$, match $d_{1, i}^u$ with $d_{2,i}^u$. Observe that the thus constructed matching~$M$ has  egalitarian cost~$c\cdot k$, because, 
   besides the~$k$ unmatched agents  in $U'$, each agent is matched with one of its most acceptable agents. To see that $M$ is also stable, it suffices to
  show that no pair of agents~$u, v \in U'$ induces a blocking
  pair. This clearly holds, since $u$ and $v$ are not adjacent in~$G$
  and thus not acceptable to each other. Thus, $M$ is a stable
  matching of egalitarian cost at most~$c\cdot k$.

  % Towards a contradiction assume that~$M$ has a
  % blocking pair~$\{u, v\}$. Clearly, for each $a_i\in A$, it holds
  % that $u \neq a_i \neq v$, because $a_i$ is matched to one of its
  % most acceptable partners. Hence, $u, v \in V$.  Since $u$ and $v$
  % are not acceptable to each other, it follows that~$u \in
  % N(v)$. Hence, since $u,v\in V(M)$, either $u \in U'$ or $v \in U'$,
  % say~$u \in U'$. This implies that~$u$ is matched to some~$a_i$,
  % meaning that $u$ gets her most acceptable agent, a contradiction to
  % the fact that $\{u, v\}$ is a blocking $M$.

  To see that a stable matching~$M$ of egalitarian cost~$c\cdot k$ induces a
  size-$k$ independent set~$U'$ for~$G$, let $U'=\{v\in U\mid M(u)\notin A\}$ be the set of agents
  in $U$, which are not assigned to some selector agent~$a_i$ by~$M$
  as a partner. Observe that each selector agent~$a_i$ is matched to
  some agent in~$U$. Otherwise, there is some agent~$b \in U$ which is
  either unmatched or matched to some other agent in~$U$ or some dummy agent. Hence
  $\{b, a_i\}$ would form a blocking pair. Thus, each selector
  agent~$a_i$ is matched to some agent in~$U$, meaning that
  $|U'| = k$. Observe also that each agent~$d_{1, i}^u$ is matched to
  $d_{2, i}^u$, because, otherwise, they would form a blocking
  pair. Hence, the only possible matches for agents in~$U'$ are their neighbors in~$G$. % Thus, the minimum cost induced by a matched agent in~$U'$ is~$c$.
  However, if two agents in~$U'$ are matched together, then $M$'s
  \egalcostn\ is in total larger than~$c\cdot k$, a contradiction. Thus, no two
  agents in $U'$ are matched together, and since $M$ is stable, no two agents in $U'$ are acceptable to each other. Thus, $U'$ is of size $k$ and induces an independent set, as required.
  % Clearly,~$|U'| \leq k$. Consider an arbitrary edge
  % $\{u, v\} \in E$, implying that the corresponding agents find each
  % other acceptable.  Since the \egalcostn{} of $M$ is $0$ we have
  % $\{u, v\} \notin M$. Moreover, since $M$ is stable, either $u$ or
  % $v$ is matched to some other agent, say~$u$ is. Again, since the
  % \egalcostn{} is~$0$, agent~$u$ is matched to some agent~$a_i$. Thus,
  % $u \in U'$ and, hence, $U'$ covers the edge~$\{u, v\}$. This implies
  % that $U'$ is a size-$k$ vertex cover.
\end{proof}
\fi
\fi

\section{Minimizing the number of blocking pairs}\label[section]{sec:SRB}
In this section, we strengthen the known result~\cite{AbBiMa2005} 
\iflong (\emph{i.e.}\ \SRB is NP-complete even for complete preferences without ties) \fi
by showing
%, through an involved reduction, 
that \SRB is W[1]-hard with respect to ``the number~$\bp$ blocking pairs'', even when each preference list has length at most five.
\iflong Note that the NP-hardness reduction in \cite{AbBiMa2005} is from the
problem of finding a maximal matching of minimum cardinality in a graph.
However, this matching problem is fixed-parameter tractable for the cardinality of the solution~\cite{Prieto2005}.
%The NP-hardness reduction in \cite{BiMaMcD2012} is from an NP-complete variant of 3-SAT.
% Moreover, the reduction is not a parameterized reduction, neither for the number of blocking pairs nor for the maximum length of a preference list\todo{ms: Even if it were, this seems irrelevant?}.
Thus, it is not clear how to adapt the proof of \citet{AbBiMa2005} to provide a parameterized reduction. \fi
\iflong

\fi
The main building block of our reduction, which is from the W[1]-hard \MIS problem (see appendix for the definition), 
is a selector gadget (\cref{gadget:selector}) that always induces at least one blocking pair and allows for many different configurations. To keep the lengths of the preference lists short we use ``duplicating'' agents (\cref{gadget:vertex}). 

\iflong
Our result excludes any $f(\bp) \cdot n^{O(1)}$-time algorithm (unless FPT${}={}$W[1]) and any $f(\bp) \cdot n^{o(\bp)}$-time algorithm (unless \ETH\ fails).
Using the same reduction, we also answer an open question by \citet[Chapter~4.6.5]{Manlove2013}, showing that minimizing the number~$\ba$ of blocking agents is NP-hard and W[1]-hard with respect to~$\ba$.
\fi

\looseness=-1 First, we discuss a vertex-selection gadget which we later use to select a vertex of the input graph into the independent set. 
The selected vertex is indicated by an agent which is matched to someone \emph{outside} of the vertex-selection gadget. 
The gadget always induces at least one blocking pair. 
\iflong
Moreover, if an agent from the gadget is not matched to one of a specified set of agents, then an arbitrary matching induces more than one blocking pair. 
An illustration is shown in the left part of \cref{fig:block-gadget}.
\begin{figure}
  \tikzstyle{pnode} = [fill=black!10]

  \centering
     \def \nodesize {17pt}
     \def \smallnodesize {15pt}

      \tikzstyle{cnode} = []
      \tikzstyle{match} = [line width = 2.6pt]
      \begin{tikzpicture}[every node/.style={draw=black,thick,circle,inner sep=0pt, font=\footnotesize}]
      \def \n {5}    
      \def \bradius {3.2cm}
      \def \radius {2.2cm}
      \def \sradius {1.3cm}
      \def \margin {12}
      \def \offs {90}
      
      \foreach \s in {0,1,2,3,4} {
        \node[draw, minimum size=\nodesize] at ({-360/\n * \s+\offs}:\radius) (a\s) {$a^\s$};
        \node[draw, minimum size=\nodesize,pnode] at ({-360/\n * \s+\offs}:\bradius) (u\s) {$u^\s$};

        \node[draw, minimum size=\smallnodesize] at ({-360/\n * \s+\offs +21}:\sradius) (c\s) {$c^\s$};
        \node[draw, minimum size=\smallnodesize] at ({-360/\n * \s+\offs - 21}:\sradius) (d\s) {$d^\s$};

         \draw (a\s) edge[normalline] (c\s);
         \draw (a\s) edge[normalline] (u\s);
         \draw (a\s) edge[normalline] (d\s);
         \draw (c\s) edge[match] (d\s);
      }

      \foreach \i /\j in {0/1,1/2,2/3,3/4,4/0} {
        \draw (a\i) edge[normalline] (a\j);
      } 

      \foreach \i /\j in {a0/a1,a3/a4,a2/u2} {
        \draw (\i) edge[match] (\j);
      }

    \end{tikzpicture}
    ~~\qquad~~
    \begin{tikzpicture}[every node/.style={draw=black,thick,circle,inner sep=1pt, font=\footnotesize}]
      \def \n {6}
      \def \nodesize {17pt}
     \def \smallnodesize {15pt}

      \tikzstyle{pnode} = [fill=black!10]
      \tikzstyle{cnode} = []
      \tikzstyle{match} = [line width = 2.6pt]
      \def \bradius {2.2cm}
      \def \radius {1.2cm}
      \def \margin {12}
      \def \offs {60}
      
      \foreach \s in {0,1,2,3,4,5} {
        \node[draw, minimum size=\nodesize] at ({-360/\n * \s+\offs}:\radius) (x\s) {$x^\s$};
      }      
       \foreach \s in {0,5} {
        \node[draw, minimum size=\nodesize, pnode] at (x\s) {$x^\s$};
      }  
      \foreach \s / \l / \nn in {0/a/a, 2/y1/{y^1}, 4/y2/{y^2}, 5/b/b} {
        \node[draw, minimum size=\nodesize] at ({-360/\n * \s+\offs}:\bradius) (\l) {$\nn$};
        \draw (x\s) edge[normalline] (\l);
      }      
        \foreach \i /\j in {0/1,1/2,2/3,3/4,4/5,5/0} {
        \draw (x\i) edge[normalline] (x\j);
      } 
      \foreach \s /\t in {x0/a,x1/x2,x3/x4,x5/b} {
        \draw (\s) edge[match] (\t);
      }       
    \end{tikzpicture}

    \label[figure]{fig:block-gadget}
    \caption{The acceptability graphs for \cref{gadget:selector,gadget:vertex} with $n'=2$ and $\delta=2$. Left: Thick lines correspond to a matching with exactly one blocking pair, $\{a^1,a^2\}$. Right: Thick lines correspond to a possible stable matching when $y^1$ and $y^2$ are matched with some agents that they prefer to $x^2$ and $x^4$, respectively.}
\end{figure}
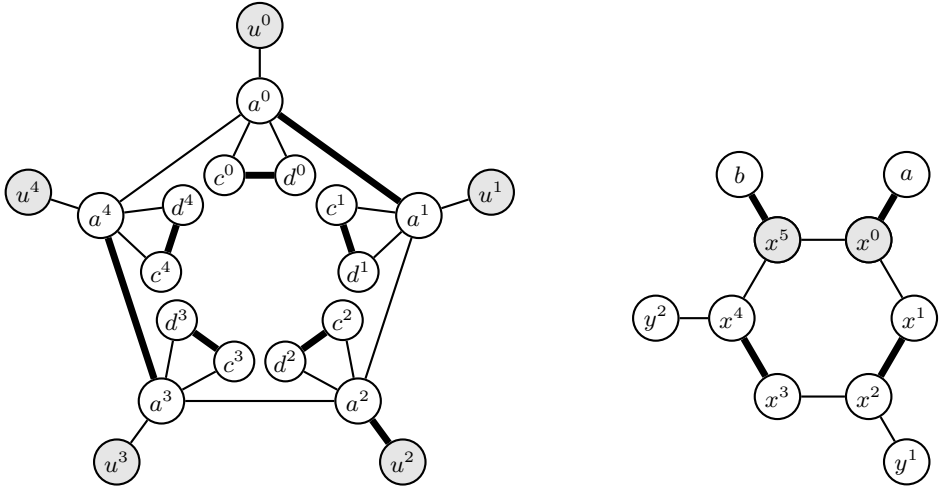
\else
An illustration is shown in the appendix. 
\fi
In the following, let $n'$ be a positive integer, and all additions and subtractions in the superscript are taken modulo~$2n'+1$:
\begin{construction}\label{gadget:selector}
Consider the following four disjoint sets~$U, A, C, D$ of $2n'+1$ agents each, where 
$A\coloneqq \{a^i \mid 0\le i \le 2n'\}$, $U\coloneqq \{u^i \mid 0\le i\le 2n'\}$, $C\coloneqq \{c^i\mid 0\le i\le 2n'\}$, 
and $D\coloneqq \{d^i\mid 0\le i \le 2n'\}$. %In the following, for some set~$X$ of agents denote by $\overrightarrow{X}$ an arbitrary but fixed order.
\ifshort
The preferences of the agents in $A\cup C \cup D$ are: 
\newcommand{\agent}{\text{agent}}
$\forall i\in$ $\{0,$ $1,\ldots, 2n'\}\colon$ \agent~$a^i\colon$ $a^{i+1} \succ a^{i-1} \succ u^i \succ c^i \succ d^i$,  $\agent$ $c^i\colon$  $d^i \succ a^i$,  \agent~$d^i\colon$  $a^i \succ c^i$.
\else
The preference lists of the agents in $A\cup C \cup D$ are as follows:
   \begin{tabular}{lll}
  $\forall i\in \{0,1,\ldots, 2n'\}\colon$ & \Agent~$a^i\colon$ & $a^{i+1} \succ a^{i-1} \succ u^{i} \succ c^i \succ d^i$,\\
  & $\Agent~c^i\colon$ & $d^i \succ a^i$,\\
  & $\Agent~d^i\colon$ & $a^i \succ c^i$.\\
\end{tabular}%
\fi
\end{construction}
The preferences of the agents in $U$ are intentionally left unspecified and we define them later when we use the gadget. Regardless of the preferences of the agents in $U$, we can verify that if no $a^i$ obtains an agent~$u^i$ as a partner,
then it induces at least two blocking pairs.

\iflong
\ifshort\begin{lemma}[\appsymb]\else
\begin{lemma}\fi
  \label[lemma]{lem:selector}
  Let $\Pot$ be a profile with agents~$A\uplus U\uplus C \uplus D$ 
  where the preferences of the agents in $A\cup C \cup D$
  obey \cref{gadget:selector}.
  Let $M$ be a matching for $\Pot$. The following holds.
\ifshort\begin{inparaenum}\else\begin{compactenum}\fi
  \item $M$ induces at least one blocking pair of the form~$\{a^{i-1}, a^{i}\}$ for some $i\in \{0,1,\ldots,2n'\}$. Moreover, if $M(a^i)\neq u^i$, then $M$ induces one more blocking pair~$p'$ with $p'\subseteq \{a^i,b^i,c^i\}$.
%  \item 
%If for each $i\in \{0,1,\ldots, 2n'\}$ it holds that $M(a^i) \neq u^i$, then $M$ induces at least two blocking pairs.
  \item Assume that there is an agent~$a^i\in A$ with $M(a^i)=u^i$.
  If (i) for each~$z \in \{1,\ldots, n'\}$ it holds that $M(a^{i+2z-1})=a^{i+2z}$,
    and (ii) for each $z'\in \{0,1,\ldots,2n'\}$ it holds that $M(c^{z'})=d^{z'}$,
  then $\{a^i,a^{i-1}\}$ is the only blocking pair that involves some agent from $A\cup C \cup D$.
\ifshort\end{inparaenum}\else\end{compactenum}\fi
\end{lemma}
\fi
\iflong
\begin{proof}
  To show the first statement, observe that $|A|$ is odd, 
  implying that there is an $i\in \{0,1,\ldots,2n'\}$ such that $M(a^i)\notin A$.
  Then, since $a^{i-1}$ ranks $a^i$ in the first position it follows that $\{a^{i-1},a^i\}$ is blocking $M$.
  Now, assume that $M(a^i)\neq u^i$.
%  To show the second statement, assume that for each $i\in \{0,1,\ldots, n'\}$ it holds that $M(a^i)\neq u^i$.
 % Since $|A|$ is odd, there is an agent~$a^i$ such that $M(a^i)\notin A$. 
%  By the reasoning for the first statement, it suffices to show that there is one more blocking pair other than~$\{a^{i-1},a^{i}\}$.
  By assumption, it holds that $M(a^i)\notin A\cup \{u^i\}$, implying that $M(a^i)\in \{\bot, c^i,d^i\}$.
  If $M(a^i)=\bot$, then since $d^i$ ranks $a^i$ in the first position, it follows that $\{a^i,d^i\}$ is also blocking~$M$.
  If $M(a^i)=c^i$, then $M(d^i)=\bot$ and $\{d^i,c^i\}$ is blocking $M$ (observe that $c^i$ ranks $d^i$ in the first position).
  If $M(a^i)=d^i$, then $M(c^i)=\bot$ and $\{c^i,a^i\}$ is blocking $M$ (observe that $a^i$ prefers $c^i$ to $d^i$).
  In any case, we find at least two blocking pairs, namely $\{a^i,a^{i-1}\}$ and another that involves two agents from $\{a^i,c^i,d^i\}$.

  It remains to show the last statement. Let $M$ be a matching for $\Pot$ with $M(a^i) = u^i$ for some~$i$ and satisfying Condition~(i) and~(ii).
  First of all, for each $z'\in \{0,1,\ldots, 2n'\}$ we observe that $c^{z'}$ already obtains its most preferred agent,
  and the only agent that $d^{z'}$ prefers to its partner~$c^{z'}$ is $a^{z'}$.
  However, $a^{z'}$ prefers its partner (which is some agent from $\{a^{z'+1}, a^{z'-1}, u^{z'}\}$) to $d^{z'}$.
  Thus, no blocking pair involves any agent from $C\cup D$.

  Now consider an agent $a^{i+2z-1}$ with $z\in \{1,2,\ldots,n'\}$.
  By the assumption on the matching~$M$,
  we know that $a^{i+2z-1}$ already obtains its most preferred agent~$a^{i+2z}$.
  Thus, no blocking pair involves any agent of the form~$a^{i+2z-1}$.
  Consider an agent~$a^{i+2z}$ with $z\in \{1,2,\ldots, n'-1\}$
  and observe that
  agent~$a^{i+2z+1}$ is the only agent that $a^{i+2z}$ prefers to its partner~$M(a^{i+2z})=a^{i+2z-1}$.
  However, as already reasoned, agent~$a^{i+2z+1}$ already obtains its most preferred agent as a partner. 
  Thus, no blocking pair involves an agent of the form~$a^{i+2z}$.
  By the reasoning for the first statement,  $\{a^{i-1},a^i\}$ is the only blocking~pair.
  \end{proof}\fi%

\looseness=-1 Next, we construct verification gadgets that ensure that no two
adjacent vertices are chosen into the independent set solution. A
straightforward idea would be to have a single agent for each vertex 
which prefers to be with its ``neighbor'' rather than with any agent from the selection gadget.
%most prefers to be with his neighbors.
If an agent and its ``neighbor'' are
forced to be matched elsewhere by the selection gadgets, they would
induce a blocking pair, exceeding the cost bound. However, this would
introduce preference lists of unbounded length. We now show how, by
configuring some additional agents in a cyclic fashion, we can reduce
the length of a preference list, while maintaining the independent set model.
\iflong
The resulting gadget is illustrated in the right part of \cref{fig:block-gadget}.
\else
The resulting gadget is illustrated in the appendix.
\fi 
Herein, let $\delta$ be a positive integer,
and all additions and subtractions in the superscript are taken modulo~$2\delta+2$.
\begin{construction}\label{gadget:vertex}Consider two disjoint sets~$X\uplus Y$ where $X=\{x^i\mid 0\le i \le 2\delta+1\}$ is a set of $2\delta+2$ agents and $Y=\{y^{i} \mid 1\le i \le \delta\}$ is a set of $\delta$ agents. 
Let $a,b$ be two agents distinct from the agents in $X\cup Y$.
The preference lists of the agents from $X$ are as follows.
\iflong
\begin{tabular}{lllll}
   & \Agent~$x^{0}\colon$ & $x^{1} \succ a \succ x^{2\delta+1}$, \\
  $\forall i\in \{1,\ldots, \delta\}\colon$ & \Agent~$x^{2i-1}\colon$ & $x^{2i}  \succ x^{2i-2}$,\\
  $\forall i\in \{1,\ldots, \delta\}\colon$ & \Agent~$x^{2i}\colon$ & $x^{2i+1} \succ y^{i} \succ x^{2i-1}$.\\
   & \Agent~$x^{2\delta+1}\colon$ & $x^{0} \succ b \succ x^{2\delta}$.\\
\end{tabular}
\else
\begin{tabular}{lllll}
   & \Agent~$x^{0}\colon$ & $x^{1} \succ a \succ x^{2\delta+1}$,& \Agent~$x^{2\delta+1}\colon$ & $x^{0} \succ b \succ x^{2\delta}$.\\
  $\forall i\in \{1,\ldots, \delta\}\colon$ & \Agent~$x^{2i-1}\colon$ & $x^{2i}  \succ x^{2i-2}$, & \Agent~$x^{2i}\colon$ & $x^{2i+1} \succ y^{i} \succ x^{2i-1}$.\\
\end{tabular}
\fi
\end{construction}
\looseness=-1 The preferences of the agents $a, b$ and those in $Y$ are intentionally left unspecified and will be defined when we use the gadget later. Regardless of the concrete preferences of agents in $Y \cup \{a, b\}$, 
we claim that the above gadget has two possible matchings such that no blocking pair involves any agent from $X$. The first one is straightforward from the definition of the preference lists: $\{\{x^{2i}, x^{2i + 1}\} \mid i \in \{0, 1, \ldots, \delta\}\}$. 
\ifshort
The second one matches $x^0$ to~$a$, $x^{2\delta+1}$ to~$b$, while keeping the remaining agents matched in some stable way.
\else
The second one is stated in the following~lemma.
\ifshort\begin{lemma}[\appsymb]\else
\begin{lemma}\fi
  \label[lemma]{lem:vertex}
 Let $\Pot$ be a profile with agents~$X\cup Y \cup \{a,b\}$ 
  where the preferences of the agents from $X$
  obey \cref{gadget:vertex}.
  Let $M$ be a matching for $\Pot$ such that $M(x^0)=a$ and $M$ does not induce any blocking pair involving an agent from~$X$.
  \iflong
 The following holds.
\else
Then, 
\fi
\ifshort\begin{inparaenum}\else\begin{compactenum}\fi
  \item \iflong For each 
\else
for each
\fi
$z \in \{1,2,\ldots,\delta\}$, it holds that 
  $M(x^{2z-1})=x^{2z}$, and $y^{z}$ prefers $M(y^{z})$ to $x^{2z}$.
  \item $M(x^{2\delta+1})=b$.
\ifshort\end{inparaenum}\else\end{compactenum}\fi  
\end{lemma}
\fi
\iflong

\begin{proof}
  The first statement can be proved by induction on $z$ with $1\le z \le \delta$.
  For $z=1$, since $M(x^0)=a$, it follows that $M(x^1)=x^2$ as otherwise $\{x^0,x^1\}$ is blocking $M$.
  Furthermore, since $x^2$ prefers $y^1$ to $x^1$ and since no blocking pair involves agent~$x^2$, it follows that $y^1$ prefers $M(y^1)$ to $x^2$.

  Now assume that $M(x^{2\delta-3})=x^{2\delta-2}$ and that $y^{\delta-1}$ prefers $M(y^{\delta-1})$ to $x^{2\delta-2}$.
  By an analogous reasoning as above we deduce that $M(x^{2\delta-1}) = x^{2\delta}$ as otherwise $\{x^{2\delta-1},x^{2\delta-2}\}$ would be blocking $M$.
  Consequently, $y^{\delta}$ must obtain an agent~$M(y^{\delta})$ that it prefers to $x^{2\delta}$.
  
  The second statement follows since $M(x^{2\delta})=x^{2\delta-1}$ and $M(x^0)=a$. 
\end{proof}
\fi
\iflong
Using \cref{lem:selector,lem:vertex}, we can prove our second main result, \cref{thm:SRB}, 
by providing a parameterized reduction from the W[1]-complete \MIS problem~\cite{fellows_parameterized_2009} parameterized by the size of the independent set~(see the proof of \cref{thm:egal-cost-fixed-w[1]-h} for the definition of \IS).
This is a special variant of \IS (also see the definition in the proof of \cref{thm:egal-cost-fixed-w[1]-h}), 
in which given a graph~$G$ on $k$ disjoint vertex subsets~$V_1,V_2,\ldots,V_k$ 
we ask whether there is a size-$k$ multi-colored independent set~$V'$ for $G$,
\emph{i.e.} an independent set which has exactly one vertex from each subset~$V_i$.
\else
Using \cref{gadget:selector,gadget:vertex}, we can prove \iflong our second main result (\cref{thm:SRB}) \else \cref{thm:SRB}.\fi
\fi
% Given a graph~$G$ with vertex set~$V$ and edge set~$E$, and a number~$k$,
% \IS{} asks whether $G$ admits a subset~$V'\subseteq V$ of $k$ pairwise non-adjacent vertices.
% Such subset of vertices is called an \emph{independent set}.

% \begin{theorem}
%   \SRB parameterized by the number of blocking pairs is W[1]-hard even in the case of complete preferences without ties. Moreover, \SRB\ does not admit an $f(k) \cdot n^{o(k)}$-time algorithm unless the \ETH\ is false.
% \end{theorem}

\iflong
\begingroup
  \def\thetheorem{\ref{thm:SRB}}
  \begin{theorem}\thmSRB
\end{theorem}
\addtocounter{theorem}{-1}
\endgroup
\fi
\ifshort
\begin{proof}[Proof sketch of \cref{thm:SRB}]
\else
\begin{proof}
  \fi
   \looseness=-1 Let $(G=(V_1,V_2,\ldots, V_k,E))$ be a \MIS instance (see appendix for ).
  Without loss of generality, assume that each vertex subset~$V_j$ has exactly $2n'+1$ vertices 
\iflong and has the following form~$V_j=\{v_j^0,v_j^1,\ldots,v_j^{2n'}\}$.   
\else
with the form $V_j=\{v_j^0,v_j^1,\ldots,v_j^{2n'}\}$.   
\fi
  Construct a \SRB instance with the following groups of agents: $U_j, A_j, B_j, C_j, D_j, F_j, W_j$, $j\in \{1,2,\ldots, k\}$,
  where $U_j$ corresponds to the vertex subset~$V_j$.
  Let $\delta_j^i$ \iflong denote \else be \fi the degree of vertex $v_j^i$.
  For each vertex~$v^i_j \in V_j$, construct $2\delta_j^i+2$ agents~$u^{i,0}_j,u_j^{i,1}, \ldots, u^{i,2\delta_j^i+1}_j$ and let $U^i_j=\{u^{i,z}_j \mid 0\le z\le 2\delta_j^i+1\}$.
  Define $U_j=\cup_{0\le i\le 2n'}U^i_j$.
  For each $(Q, q) \in \{(A, a), (B, b), (C, c), (D, d), (F, f), (W, w)\}$ and for each $i\in \{1,2,\ldots,k\}$,
  the set~$Q_j\coloneqq \{q_j^i \mid 0\le i \le 2n'\}$ consists of $2n'+1$ agents.
  The preference lists of the agents in $U^i_j$ obey the verification gadget constructed in \cref{gadget:vertex}.
  Formally, for each $j \in \{1, \ldots, k\}$ and each $i \in \{0, 1, \ldots, 2n'\}$ we introduce a \emph{verification gadget for $v_j^i$} as in \cref{gadget:vertex} where we set $\delta = \delta_j^i$,  $x^z=u^{i,z}_j$, $0\le z\le 2\delta_j^i+1$, $a=a^i_j$, and $b=b^i_j$.
  The agents from $Y$ correspond to the neighbors of $v^i_j$:
  For each neighbor~$v^{i'}_{j'}$ of $v^i_j$ 
  we pick a not-yet-set agent~$y^z$ in the verification gadget for $v^{i}_{j}$
  and a not-yet-set agent~$y^{z'}$ in the verification gadget for $v^{i'}_{j'}$,
  and define $y^z=u^{i',2z'}_{j'}$ and $y^{z'}=u^{i,2z}_{j}$.

  \iflong
  For instance, if vertex~$v^0_1$ and $v^1_2$ are adjacent,
  and the agent~$y^2$ for $v^0_1$ and the agent~$y^1$ for $v^1_2$ are not yet set,   
  then the agent~$y^2$ for $v^0_1$ could be $u^{1,2}_2$
  and the agent~$y^1$ for $v^1_2$ could be $u^{0,4}_1$.
  In this way, the preference list of $u^{0,4}_1$ will have the form $u^{0,5}_1\succ u^{1,2}_2 \succ u^{0,3}_1$.
  The preference list of $u^{1,2}_2$ will have the form $u^{1,3}_2\succ u^{0,4}_1 \succ u^{1,1}_2$.
  \fi

  For each $j \in \{1, \ldots, k\}$, the preference lists of $A_j\cup C_j \cup D_j \cup \{u^{i,0}_j\mid 0\le i \le 2n'\}$ obey \cref{gadget:selector}. 
  Formally, for each $j \in \{1, \ldots, k\}$ we introduce a vertex-selection gadget as in \cref{gadget:selector} and for each $i \in \{0,1,\ldots, 2n'\}$ we set $a^i=a^i_j$, $c^i=c^i_j$, $d^i=d^i_j$,
  and $u^i=u^{i,0}_j$.
  Analogously, 
  \iflong 
  for each $j \in \{1, \ldots, k\}$, the preference lists of $B_j\cup F_j \cup W_j \cup \{u^{i,2\delta_j^i+1}_j\mid 0\le i \le 2n'\}$ obey \cref{gadget:selector}.
  Formally, 
  \fi for each $j \in \{1, \ldots, k\}$ we introduce a vertex-selection gadget~for  $B_j\cup F_j \cup W_j \cup \{u^{i,2\delta_j^i+1}_j\mid 0\le i \le 2n'\}$: For each $i \in \{0,1,\ldots, 2n'\}$ we set $a^i=b^i_j$, $c^i=f^i_j$, $d^i=w^i_j$,
  and $u^i=u^{i,2\delta_i^j+1}_j$.
  To complete the construction, we set the upper bound on the number of blocking pairs as~$\bp=2k$.
  \iflong
  
 We show that $G=(V_1,V_2,\ldots, V_k, E)$ is a yes-instance of \MIS if and only if the constructed profile admits a matching with at most $\bp\coloneqq 2k$ blocking pairs.

For the ``if'' part, assume that $M$ is a matching with at most~$2k$ blocking pairs.
By the first statement of \cref{lem:selector} it follows that
for each $j\in \{1,2,\ldots, k\}$, matching~$M$ induces at least two blocking pairs~$p^1_j$ and $p^2_j$ of the forms~$p^{1}_j=\{a^{i-1}_j,a^{i}_j\}$ and $p^{2}_j = \{b^{i'-1}_j,b^{i'}_j\}$ for some $i,i'\in \{0,1,\cdots,2n'\}$.

By the second part of the first statement of % \todo{ms: Wieso ist die Vorbedingung des zweiten Statements erfüllt?}second statement of
\cref{lem:selector}, the agent~$a^i_j$ that is involved in the blocking pair $p^1_j$ must be matched to $u^{i,0}_j$ as otherwise there will be more than $2k$ blocking pairs. 
We claim that $V'=\{v^i_j \mid \{a^i_j, u^{i,0}_j\} \in M\}$ is a size-$k$ multi-colored independent set.
Obviously, $|V'|=k$.

Suppose, for the sake of contradiction, that $V'$ contains two adjacent vertices. 
Let these vertices be $v^{i}_j$ and $v^{i'}_{j'}$ for two distinct $j,j'\in \{1,2,\ldots, k\}$ and some $i,i'\in \{0,1,\ldots,2n'\}$.
By the definition of $V'$, we have that $\{a^i_j, u^{i,0}_j\}, \{a^{i'}_j, u^{i',0}_{j'}\} \in M$.
By the first statement of \cref{lem:vertex}, 
for each $z\in \{1,\ldots, \delta^i_j\}$ it holds that 
$\{u^{i,2z-1}_j, u^{i,2z}_j\}\in M$, and 
for each $z' \in \{1,\ldots, \delta^{i'}_{j'}\}$ it holds that
$\{u^{i',2z'-1}_{j'}, u^{i',2z'}_{j'}\} \in M$.
However, since $v^{i}_j$ and $v^{i'}_{j'}$ are adjacent, by the preference lists under \cref{gadget:vertex} %\todo{ms: Nicht eher ``since $v_j^i$ and $v_{j'}^{i'}$ are adjacent''?}
for $U^i_j$ and $U^{i'}_{j'}$, 
there are two agents~$u^{i,2z}_j$ and $u^{i',2z'}_{j'}$ (for some $z \in \{0,1,\ldots, \delta_i^j\} ,z' \in \{0,1,\ldots,\delta_{i'}^{j'}\}$) such that $u^{i,2z}_j$ prefers $u^{i',2z'}_{j'}$ to  its partner~$u^{i,2z-1}_{j}$
and  $u^{i',2z'}_{j'}$ prefers $u^{i,2z}_{j}$ to  its partner~$u^{i',2z'-1}_{j'}$.
This implies that $\{u^{i,2z}_j,u^{i',2z'}_{j'}\}$ is blocking $M$, that is, $M$ induces more than $2k$ blocking pairs---a contradiction. Hence, indeed $V'$ is an independent set.

\newcommand{\BP}{\textsf{BP}}
For the ``only if'' part, assume that $V'\subset V$ is a multi-colored independent set of size~$k$.
We claim that $\BP=\{\{a^{i-1}_j, a^{i}_j\}, \{b^{i-1}_j, b^{i}_j\}\mid v^i_j\in V'\}$ consists of all $2k$ blocking pairs of the matching~$M$ defined as follows.
%Take $\{a^{i-1}_j, a^i_j\} \in M$ and $\{b^{i-1}_j, b^i_j\} \in M$, where $v^i_j \in V'$.
For each color~$j\in \{1,2,\ldots,k\}$ and for each $i \in \{0,1,\ldots, 2n'\}$ do the following; all additions and subtractions in subscripts of the agents in the vertex-selection gadgets are taken modulo $2n'+1$.
\begin{compactitem}
\item Set $M(c^i_j)=d^i_j$. If $v^i_j\in V'$ (\emph{i.e.}\ the vertex belongs to the multi-colored independent set), 
then set $M(u^{i,0}_j)=a^i_j$ and $M(u^{i,2\delta_j^i+1}_j)=b^i_j$, 
for each $z\in \{1,2,\ldots,n'\}$ set $M(a^{i+2z-1}_j)=a^{i+2z}_j$ and $M(b^{i+2z-1}_j)= b^{i+2z}_j$.
\end{compactitem}
Note that this in particular defines matchings for all vertex-selection gadgets. 
\iflong By
the third statement in \cref{lem:selector}, the only blocking pairs
introduced so far are for each $v^i_j\in V'$ the
pairs~$\{a^{i-1}_j, a^i_j\}$ and~$\{b^{i-1}_j, b^i_j\}$.
\fi
\begin{compactitem}
\item If $v^i_j \notin V'$, then for each
  $z \in \{0, 1, \ldots, \delta_j^i\}$ set
  $M(u^{i,2z}_j) = u^{i,2z + 1}_j$. 

\item Finally, for each $v^i_j\in V'$ and each
  $z\in \{1,2,\ldots,\delta_j^i\}$ set $M(u^{i,2z-1}_j)=u^{i,2z}_j$.
\end{compactitem}

\noindent We claim that $\BP$ as defined above consists of all blocking pairs of~$M$.
Towards a contradiction, suppose that $M$ induces a blocking pair~$p\notin \BP$.
As mentioned, by the third statement of \cref{lem:selector} and the
construction of our matching~$M$, besides the blocking pairs in $\BP$, 
the only agents that could form a blocking pair are from verification gadget(s).
We distinguish two cases.
In the first case the two agents in $p$ are in the verification gadget for one and the same vertex~$v^i_j\in V_j$ for some $j \in \{1,2,\ldots, k\}$ and $i\in \{0,1,\ldots, 2n'\}$.
Then, $v^i_j\in V'$ because for each two agents that correspond to the same vertex~$v'\in V\setminus V'$ (which is not in the independent set solution~$V'$)
and are acceptable to each other, one of them already obtains its most preferred agent. Assume that $p=\{u^{i,z}_j, u_j^{i,z'}\}$ for some independent set vertex~$v^i_j\in V'$ such that $u^{i,z}_j$ and $u_j^{i,z'}$ are acceptable to each other but not matched together.
If $z, z'\notin \{0, 2\delta^i_j, 2\delta^i_j+1\}$, 
then $p$ cannot form a blocking pair as one of the agents in $p$ already obtains its most preferred agent.
This implies that $\{z,z'\}=\{0,2\delta^{i}_j+1\}$ or $\{z,z'\}=\{2\delta^i_j, 2\delta^i_j+1\}$
because agents~$u^{i,0}_j$ and $u^{i,2\delta^i_j}_j$ do not find each other acceptable.
However, $u^{i,0}_j$ prefers its partner~$a^i_j$ to $u^{i,2\delta^{i}_j+1}_j$
and  $u^{i,2\delta^i_j+1}_j$ prefers its partner~$b^i_j$ to $u^{i,2\delta^{i}_j}_j$.
Thus, we deduce that no blocking pair involves two agents that correspond to the same vertex.

It remains to consider the case when the two agents in $p$ are in verification gadgets of two different vertices~$v^i_j, v^{i'}_{j'}\in V$ with $v^i_j\neq v^{i'}_{j'}$.
Since only agents that correspond to two adjacent vertices could find each other acceptable,
it follows that $v^i_j$ and $v^{i'}_{j'}$ are adjacent.
By the preference lists of the verification gadgets for $v^i_j$ and $v^{i'}_{j'}$,
it also follows that $p=\{u^{i,2z}_j, u^{i',2z'}_{j'}\}$ for some $z,z'\in\{1,2,\ldots,k\}$.
Since $p$ is blocking~$M$ it follows that $M(u^{i,2z}_j)=u^{i,2z-1}_j$ and $M(u^{i,2z}_j)=u^{i,2z-1}_j$.
Moreover, by the construction of~$M$ it must hold that 
$\{u^{i,0}_j, a^{i}_j\}, \{u^{i',0}_{j'}, a^{i'}_{j'}\}\in M$ for some $j, j'\in \{1,2,\ldots, k\}$,
meaning that $v^i_j, v^{i'}_{j'}\in V'$---a contradiction to $V'$ being an independent set.

Assuming \ETH, if our problem would have an $f(\bp)\cdot n^{o(\bp)}$-time algorithm,
then by our reduction, \MIS would also admit an $g(k)\cdot n^{o(k)}$-time algorithm, where $f$ and $g$ are two computable functions---a contradiction to~\cite[Corollary 14.23]{CyFoKoLoMaPiPiSa2015}.
\else
The correctness proof is deferred to the appendix.
\fi
\end{proof}

\iflong{}
The reduction given in the proof of \cref{thm:SRB} shows  
\else{} \noindent The reduction given in the proof of \cref{thm:SRB} shows  
\fi
\iflong{}that the lower-bound on the number~$\bp$ of blocking pairs given by \citet[Lemma~4]{AbBiMa2005} is tight. 
\else{}that the lower-bound on $\bp$ given by \citet[Lemma~4]{AbBiMa2005} is tight. 
\fi
The reduction also answers an open question
by \citet[Chapter~4.6.5]{Manlove2013} pertaining to the complexity of the
following related problem. In \SRA, we are given a preference profile and an integer~$\ba$, 
and we want to know whether there is a matching with at most~$\ba$ \emph{blocking} agents, that is, agents involved in blocking pairs.%. Herein, an agent is \emph{blocking} if it is involved in a blocking pair.

\iflong
First of all, we observe the following for the vertex-selection gadget given in \cref{gadget:selector}.
\begin{lemma}
  \label[lemma]{lem:selector-ba}
  Let $\Pot$ be a preference profile with agents~$A\uplus U\uplus C \uplus D$ 
  where the preferences of the agents in $A\cup C \cup D$
  obey \cref{gadget:selector}.
  Let $M$ be a matching for $\Pot$. The following~holds.
\begin{compactenum}
  \item $M$ induces at least two blocking agents~$a^{i-1}, a^{i}$ for some $i\in \{0,1,\ldots,2n'\}$.
  \item  If for each $i\in \{0,1,\ldots, 2n'\}$ it holds that $M(a^i) \neq u^i$,
  then $M$ induces at least three blocking agents.
  \item Assume that there is an agent~$a^i$, $i\in \{0,1,\ldots, 2n'\}$, with $M(a^i)=u^i$.
  If (i) for each~$z \in \{1,\ldots, n'\}$ it holds that $M(a^{i+2z-1})=a^{i+2z}$,
    and (ii) for each $z'\in \{0,1,\ldots,2n'\}$ it holds that $M(c^{z'})=d^{z'}$,
  then agents~$a^{i-1}$ and $a^{i}$ are the only two blocking agents from $A\cup C \cup D$.
\end{compactenum}
\end{lemma}

\begin{proof}
  For the first statement, we observe that the first statement of  \cref{lem:selector} implies that there is an integer~$i\in \{0,1,\ldots, 2n'\}$ with $\{a^{i-1}, a^{i}\}$ being a blocking pair.
  Thus, $a^{i-1}$ and $a^{i}$ are two blocking agents.
  
  Similarly, by the reasoning for the second statement of \cref{lem:selector}, 
  if no agent~$a^i$ is matched with $u^i$, then there are at least two blocking pairs, which share at most one blocking agent. 
  Thus, there are at least three blocking agents.
  
  The last statement follows by the last statement of \cref{lem:selector}.
\end{proof}
\fi

\ifshort\begin{corollary}[\appsymb]\else
\begin{corollary}\fi\label[corollary]{cor:srba-w[1]-h}
  Let $n$ be the number of agents and $\ba$ be the number of blocking agents.
  Even when each input preference list has length at most five and has no ties, \SRA is NP-hard and
  W[1]-hard with respect to $\ba$. 
  \SRA for preferences without ties is solvable in $O(2^{\ba^2}\cdot n^{\ba+2})$~time.
\end{corollary}
\iflong

\begin{proof}
  To show the hardness results, we use the same reduction as in the proof of \cref{thm:SRB} and we set the number of blocking agents allowed to be $\ba=2k$.
  As in the proof of \cref{thm:SRB}, now using \cref{lem:selector-ba}, if
  the \SRA instance is a yes-instance, then the \IS\ instance is also a yes-instance. 
  In the other direction, observe that the matching
  constructed from a size-$k$ independent set solution in
  \cref{thm:SRB} induces exactly $2k = \ba$ blocking agents.

  In the remainder of the proof, we provide an algorithm for our problem and show that the running time is $O(2^{\ba^2}\cdot n^{\ba+2})$.
  This algorithm uses as a subprocedure the algorithm provided by \citet{AbBiMa2005} that, 
  given a set~$B$ of pairs of agents, checks in $O(n^2)$ time 
  whether there is matching~$N$ such that each pair in $B$ is blocking $N$ and no other pair is blocking $N$.
  
  Assume that there is a matching~$M$ with exactly $\ba$ blocking agents.
  First, we guess the set of these $\ba$ agents that form some blocking pairs of $M$ and denote this set as $A^*=\{a_1,a_2, \dots, a_{\ba}\}$.
  Now, observe that any blocking pair must only involve agents from~$A^*$.
  Thus, we guess the blocking pairs of $M$ that only involve agents from~$A^*$.
  Finally, we call the algorithm by \citet{AbBiMa2005} to check whether the guess is correct.

  It is straightforward to see that there is a matching~$M$ with exactly $\ba$ blocking agents if and only if one of the guesses gives a subset of pairs of agents that are exactly the blocking pairs of $M$.
  Now observe that there are $O(n^{\ba})$ guesses of the subset of blocking agents with size $\ba$ and for each of these subset of blocking agents there are $O(2^{\ba^2})$ guesses for a corresponding subset of blocking pairs.
  Since for each guessed subset of blocking pair, we invoke the algorithm by \citet{AbBiMa2005} which runs in $O(n^2)$ time,
  our algorithm has a running time of $O(2^{\ba^2}\cdot n^{\ba+2})$.
\end{proof}
\fi

\iflong
\section{Conclusion and outlook}\label[section]{sec:conclusion}

We showed that \ESR and \SRB, though both NP-hard in the classical complexity point of view, 
behave completely differently in a parameterized perspective.
In particular, we showed that \ESR is fixed-parameter tractable with respect to the \egalcostn~$\egalcost$
while \SRB with bounded preference length is W[1]-hard with respect to the number~$\bp$ of blocking pairs.

Our work leads to some open questions.
First, we showed that for preferences without ties, \ESR admits a size-$O(\egalcost^2)$ kernel.
It would thus be interesting to see whether \ESR also admits a polynomial kernel when ties are present.
Second, it would be interesting to see whether the running time in \cref{thm:ESR}
is~tight.
% Third, our W[1]-hard for \SRB holds even when the preference lists do not have ties and have length at most five. 
% An interesting future question is to determine the parameterized complexity when the lengths of the preference lists are bounded by four.
\fi

%\bibliographystyle{plainurl}% the recommended bibstyle
% \bibliographystyle{abbrvnat}
% \bibliography{bib}

\end{document}